\documentclass[preprint,12pt]{elsarticle}

\usepackage{amssymb}
\usepackage{amsthm}

\usepackage{comment}
\usepackage{graphicx,amssymb,amsmath}
\usepackage{gensymb}
\usepackage{longtable}
\usepackage{array,booktabs,enumitem}
\newcolumntype{P}[1]{>{\endgraf\vspace*{-\baselineskip}}p{#1}}

\makeatletter
\g@addto@macro\bfseries{\boldmath}
\makeatother

\usepackage[explicit]{titlesec}
\titleformat{\paragraph}[runin]{\normalfont\bfseries\itshape}{\theparagraph}{1em}{#1.}

\makeatletter
\def\corref#1{\edef\cnotenum{\elsRef{#1}}%
	\edef\@corref{\ifcase\cnotenum\or
		$\dagger$\or$\dagger\dagger$\fi\hskip-1pt}}

\def\cortext[#1]#2{\g@addto@macro\@cornotes{%
		\refstepcounter{cnote}\elsLabel{#1}%
		\def\thefootnote{\ifcase\thecnote\or$\dagger$\or
			$\dagger\dagger$\fi}%
		\footnotetext{#2}}}
\makeatother

\makeatletter
\def\th@plain{%
	\thm@notefont{}
	\itshape 
}
\def\th@definition{%
	\thm@notefont{}
	\normalfont 
}
\makeatother

\makeatletter
\def\ps@pprintTitle{%
	\let\@oddhead\@empty
	\let\@evenhead\@empty
	\def\@oddfoot{}%
	\let\@evenfoot\@oddfoot}
\makeatother

\graphicspath{{./figs/}}

\journal{JoCG}

\newtheorem{thm}{Theorem}[section]
\newtheorem{lem}[thm]{Lemma}
\newproof{pf}{Proof}
\theoremstyle{definition}
\newtheorem{definition}{Definition}[section]

\begin{document}

\begin{frontmatter}
	
	
	
	\title{Computing Feasible Trajectories for an Articulated Probe in Three Dimensions\tnoteref{label1}}
	\tnotetext[label1]{A preliminary version of this work was presented at the 31st Annual Canadian Conference on Computational Geometry.}
	
	
	\author{Ovidiu Daescu}
	\ead{ovidiu.daescu@utdallas.edu}
	
	\author{Ka Yaw Teo\corref{cor1}}
	\ead{ka.teo@utdallas.edu}
	
	\cortext[cor1]{Corresponding author}
	\address{Department of Computer Science, University of Texas at Dallas, Richardson, TX, USA.}

\begin{abstract}
Consider an input consisting of a set of $n$ disjoint triangular obstacles in $\mathbb{R}^3$ and a target point $t$ in the free space, all enclosed by a large sphere $S$ of radius $R$ centered at $t$.
An articulated probe is modeled as two line segments $ab$ and $bc$ connected at point $b$.
The length of $ab$ can be equal to or greater than $R$, while $bc$ is of a given length $r \leq R$.
The probe is initially located outside $S$, assuming an \emph{unarticulated} configuration, in which $ab$ and $bc$ are collinear and $b \in ac$.
The goal is to find a feasible (obstacle-avoiding) probe trajectory to reach $t$, with the condition that the probe is constrained by the following sequence of moves -- a straight-line insertion of the unarticulated probe into $S$, possibly followed by a rotation of $bc$ at $b$ for at most $\pi/2$ radians, so that $c$ coincides with $t$.

We prove that if there exists a feasible probe trajectory, then a set of \emph{extremal} feasible trajectories must be present.
Through careful case analysis, we show that these extremal trajectories can be represented by $O(n^4)$ combinatorial events.
We present a solution approach that enumerates and verifies these combinatorial events for feasibility in overall $O(n^{4+\epsilon})$ time using $O(n^{4+\epsilon})$ space, for any constant $\epsilon > 0$.
The enumeration algorithm is highly parallel, considering that each combinatorial event can be generated and verified for feasibility independently of the others.
In the process of deriving our solution, we design the first data structure for addressing a special instance of circular sector emptiness queries among polyhedral obstacles in three dimensional space, and provide a simplified data structure for the corresponding emptiness query problem in two dimensions.
\end{abstract}

\end{frontmatter}

\section{Introduction}
In this paper, we address the three-dimensional (3D) version of the \emph{articulated probe trajectory planning problem} introduced by Teo, Daescu, and Fox \cite{teo20traj}.
We are given a 3D workspace containing a set $P$ of $n$ interior disjoint triangular obstacles and a target point $t$ in the free space, all located within a large sphere $S$ of radius $R$ centered at $t$ (see Figure \ref{fig_traj}).
An articulated probe is modeled as two line segments $ab$ and $bc$ joined at point $b$.
Line segment $ab$ has a length of $R$ or greater, while $bc$ has a fixed length $r \in (0, R]$.
Line segment $bc$ may rotate around $b$.

\begin{figure}[h]
	\centering
	\includegraphics[scale=0.24]{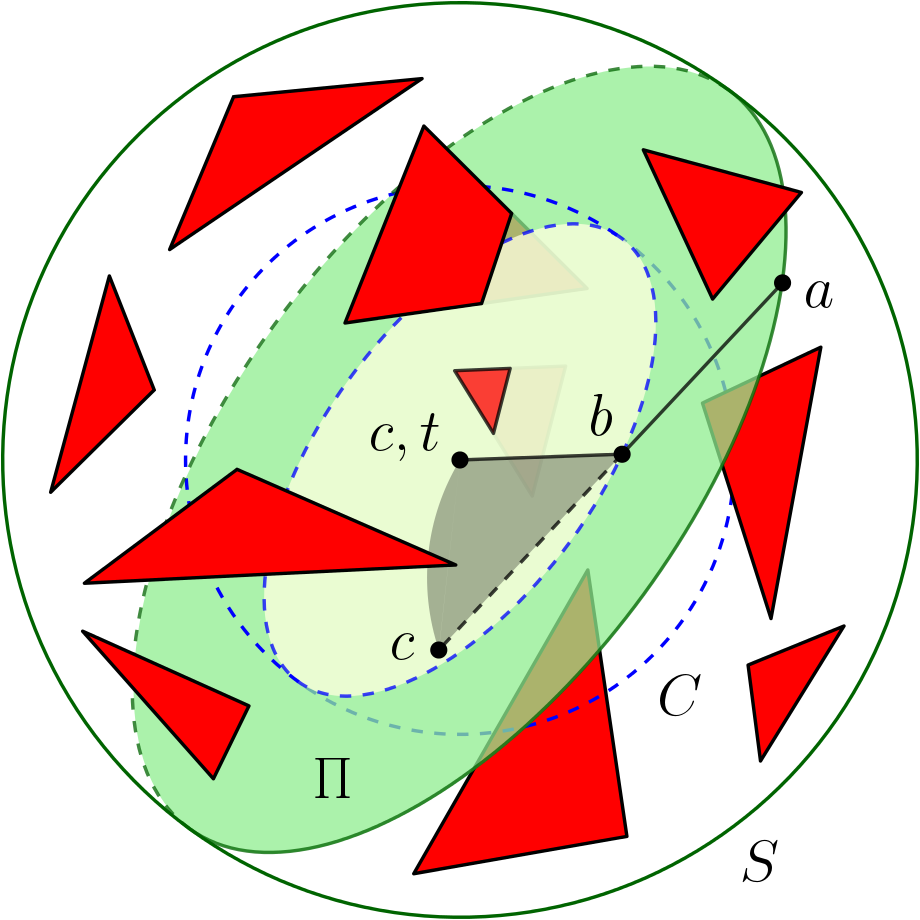}
	\caption{Articulated probe trajectory in 3D.  After inserting line segment $abc$ into sphere $S$, in order to reach target point $t$, line segment $bc$ may be required to rotate from its intermediate position (dashed line segment) to its final position (solid line segment).}
	\label{fig_traj}
\end{figure}

The probe is said to have an \emph{unarticulated} configuration if line segments $ab$ and $bc$ are collinear, and $b \in ac$.
Otherwise, the probe has an \emph{articulated} configuration -- that is, $bc$ has been rotated around $b$, and $bc$ is no longer collinear with $ab$.

The probe begins its trajectory in an unarticulated configuration outside sphere $S$.
The unarticulated probe, represented by straight line segment $abc$, is at first inserted along a straight line into $S$.
After completing the insertion, line segment $bc$ may be rotated around point $b$ up to $\pi/2$ radians in order to reach the target point $t$.
Hence, the final configuration of the probe could be either unarticulated or articulated.
The \emph{intermediate} configuration of the probe is the unarticulated configuration after the insertion and before the rotation.
Without loss of generality, point $a$ is assumed to be located on $S$ after the insertion.

A \emph{feasible} probe trajectory consists of an initial insertion of straight line segment $abc$ into sphere $S$, possibly followed by a rotation of line segment $bc$ around point $b$, such that point $c$ ends at the target point $t$, while avoiding the obstacles in the process of insertion and rotation.
The objective of the problem is to find a feasible probe trajectory, if one exists.

As illustrated in Figure \ref{fig_traj}, a feasible probe trajectory is planar -- that is, its motion always lies in a plane $\Pi$ passing through the target point $t$.
Since line segment $bc$ may only rotate as far as $\pi/2$ radians around point $b$, for any feasible probe trajectory, point $b$ is the first (and only) intersection between line segment $ab$ and the sphere $C$ of radius $r$ centered at the target point $t$.
When line segment $bc$ rotates around point $b$, the area swept by $bc$ is a sector $\sigma$ of a circle $D$ in plane $\Pi$ (i.e., a portion of a disk enclosed by two radii and a circular arc), where $D$ is of radius $r$ and centered at a point on sphere $C$. 
Circular sector $\sigma$ always has an endpoint of one of its bounding radii located at the target point $t$.

\subsection*{Motivation}
Besides its pointed relevance to robotics, the outlined problem arises particularly from planning for minimally invasive surgeries.
In fact, surgical instruments that can be modeled by our simple articulated probe are already in clinical use (e.g., da Vinci EndoWrist by Intuitive Surgical), given their enhanced capability in reaching remote targets while circumventing surrounding critical structures \cite{simaan18medical}.
In our problem setting, a human body cavity can be viewed as (a subset of) workspace $S$, and any critical organ/tissue can be represented by using a triangle mesh.
Despite its importance and relevance, the problem has never been investigated in three dimensions from a theoretical viewpoint, and only a handful of results in two dimensions have been reported \cite{daescu20char,teo20traj}.

\subsection*{Related work}
The motion of a linkage or kinematic chain -- that is, an assembly of rigid links and joints -- has been investigated in many different aspects, including the basic geometric properties and topology of linkages (e.g., reconfiguration and locked decision) \cite{biedl01locked,connelly17geom} as well as application-driven questions related to linkage design and motion planning \cite{choset05principles,lavalle06plan}.

Approximation methods based on sampling \cite{kavraki96probabilistic,lavalle01randomized} and subdivision \cite{brooks85subdivision,donald84motion,zhu90constraint} are commonly used (and widely regarded as efficient) for motion planning involving complex multi-link robots with high degrees-of-freedom.
However, a robotic linkage can often be associated with spatial (link or joint) constraints, resulting in a restricted free configuration space, which has proven to be challenging for inexact sampling-based algorithms to find valid solutions \cite{yakey01randomized}.
In addition, sampling- and subdivision-based approaches are only complete within their inherent limits -- that is, probabilistically and resolution-wise, respectively.

In contrast to previous work on polygonal linkages, which are generally allowed to rotate unrestrictedly at their joints while moving from a start to a final configuration \cite{choset05principles,connelly17geom,lavalle06plan}, our study is concerned with a simple articulated two-bar linkage subjected to a specific sequence of moves -- namely, a straight-line insertion of the linkage followed by a rotation at its joint.
Moreover, one of the links is assumed to have a length $\geq R$ (i.e., radius of spherical workspace $S$).
Without resorting to approximation, we develop an exact geometric-combinatorial approach to solve the proposed path planning problem.

Teo, Daescu, and Fox \cite{teo20traj} originally proposed the aforementioned trajectory planning problem in two dimensions (2D), and they presented an $O(n^2 \log n)$-time, $O(n \log n)$-space algorithm for finding a feasible trajectory amidst $n$ line segment obstacles.
The algorithm was based on computing extremal trajectories that are tangent to one or two obstacle vertices.
In addition, they showed that, for any constant $\delta > 0$, a feasible trajectory of a clearance $\delta$ from the obstacles can be determined in $O(n^2 \log n)$ time using $O(n^2)$ space.
These algorithmic approaches were also extended to the case of $h$ polygonal obstacles, where a feasible trajectory can be found in $O(n^2 + h^2 \log n)$ time using $O(n \log n)$ space when no clearance is required, or $O(n^2)$ space otherwise.

Daescu and Teo \cite{daescu20char} later addressed a variant of the trajectory planning problem where the length $r$ of the end segment $bc$ is not fixed.
They demonstrated that, for any constant $\epsilon > 0$, the shortest length $r > 0$ for which a feasible trajectory exists can be determined in $O(n^{2+\epsilon})$ time using $O(n^{2+\epsilon})$ space, and at least one such trajectory can be reported with the same time/space complexity.
The proposed algorithm can be extended to reporting all feasible values of $r$ in $O(n^{5/2})$ time using $O(n^{2+\epsilon})$ space.
In the same work, Daescu and Teo also showed that, for a given $r$, the feasible solution space for the trajectory planning problem can be characterized by a simple-curve arrangement of complexity $O(k)$, and the arrangement can be constructed in $O(n \log n + k)$ time using $O(n + k)$ space, where $k = O(n^2)$ is the number of vertices of the arrangement.

\subsection*{Results and contributions}
In this paper, we describe an algorithm that computes a feasible probe trajectory in 3D, if one exists, in $O(n^{4+\epsilon})$ time using $O(n^{4+\epsilon})$ space, for any constant $\epsilon > 0$.
First, we prove that if there exists a feasible probe trajectory, then some \emph{extremal} feasible trajectories must be present.
An extremal trajectory is characterized by its intersections or tangencies with a combination of obstacle edges, vertices, and/or surfaces.
Through careful case analysis, we show that these extremal trajectories can be represented by $O(n^4)$ combinatorial events.
Our algorithm is based on enumerating and verifying these combinatorial events for feasibility.
As an alternative, an $O(n^5)$-time algorithm with $O(n)$-space usage is achievable by performing a simple $O(n)$-check on each of the $O(n^4)$ events.

The main new difficulty that arises in this study, when compared to the 2D case in \cite{teo20traj}, is the possibility that for an extremal trajectory in 3D, a segment of the probe may intersect an obstacle edge at an interior point, and an obstacle endpoint may be incident to the interior of the circular sector representing the area swept by the end segment $bc$ of the probe.
As a consequence, in order to address the 3D case, procedures and data structures different from those previously used in \cite{teo20traj} are required, particularly for extremal trajectory computation and circular sector emptiness queries.

While deriving our solution approach, we develop the first data structure for solving a special case of the circular sector emptiness query problem in 3D, where the query circular sector has a fixed radius $r$ and an endpoint of its arc located at fixed point $t$.
We present a data structure of size $O(n^{4+\epsilon})$ for answering a query of the sort in $O(\log n)$ time.
When mapped to the plane, this result yields a new data structure for solving the corresponding circular sector emptiness query problem in 2D.
Our new $\mathbb{R}^2$ query data structure simplifies the two-part approach formerly proposed in \cite{teo20traj} while maintaining the same time and space complexity.
Sharir and Shaul \cite{sharir11semialgebraic} proposed a solution based on semialgebraic range searching for circular cap (i.e., portion of a circle cut off by a line and larger than a semidisk) emptiness queries in 2D.
These circular sector emptiness queries, in 2D and 3D, are considered to be of independent interest.
To the best of the authors' knowledge, there has been no published data structure for general circular sector emptiness queries in 3D or even 2D.

The remainder of the paper is organized as follows.
In Section \ref{extr}, we prove the presence of extremal feasible probe trajectories.
In Section \ref{comp_valid}, we describe an algorithm for computing extremal probe trajectories and verifying them for feasibility.
We conclude in Section \ref{conc} by providing a summary of our results and a few related open questions.

\section{Extremal feasible trajectories}
\label{extr}
In this section, we prove that if there exists a feasible probe trajectory, then a set of extremal feasible trajectories must also be present.
We begin by defining a series of terminologies used in our ensuing discussion.

Let $\ell$ denote a line segment of the probe.
In addition, let $\sigma$ and $\gamma$ denote a circular sector (i.e., area swept by $bc$, as previously described) and its arc, respectively.
Let $\tau$ be a triangular obstacle of $P$ in $\mathbb{R}^3$, and let $e$ denote en edge of $\tau$.
Without loss of generality, assume that $\tau$ is not co-planar with $t$, and $\ell$ is not parallel to $e$.
Line segment $\ell$ may intersect $e$ at an endpoint of $e$, an interior point of $e$, or none.

\begin{definition}[Support vertex of $\ell$]
	If $\ell$ intersects $e$ at an endpoint $p$ of $e$, then $p$ is called a \emph{support vertex} of $\ell$.
\end{definition}

\begin{definition}[Support edge of $\ell$]
	If $\ell$ intersects $e$ at an interior point of $e$, then $e$ is called a \emph{support edge} of $\ell$.
\end{definition}

\noindent
Thus, a \emph{support} of $\ell$ can be either a support vertex or a support edge.
Suppose, without loss of generality, that $\sigma$ and $e$ lie in different planes.

\begin{definition}[Support vertex of $\sigma$]
	If $\sigma$ contains an endpoint $p$ of $e$, then $p$ is a \emph{support vertex} of $\sigma$.
\end{definition}

\begin{definition}[Support edge of $\gamma$]
	If $\gamma$ is tangent to $e$ at an interior point of $e$, then $e$ is a \emph{support edge} of $\gamma$.
\end{definition}

\begin{definition}[Support surface of $\gamma$]
	If $\gamma$ is tangent to the surface of $\tau$, then the surface is a \emph{support surface} of $\gamma$.
\end{definition}

\noindent
For an illustration of the various types of supports just described, see Figure \ref{fig_support}. 

\begin{figure}[h]
	\centering
	\includegraphics[scale=0.19]{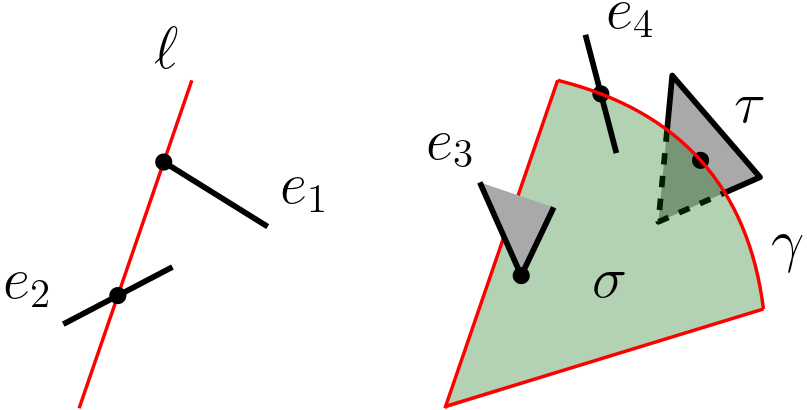}
	\caption{Support vertex, edge, and surface.
		An endpoint of edge $e_1$ is a support vertex of line segment $\ell$.
		Edge $e_2$ is a support edge of $\ell$.
		An endpoint of $e_3$ is a support vertex of circular sector $\sigma$.
		Edge $e_4$ supports circular arc $\gamma$ at an interior point of $e_4$.
		Triangle $\tau$ is a support surface of $\gamma$.}
	\label{fig_support}
\end{figure}

A line segment is \emph{extremal} (or \emph{isolated}) with respect to a set of support vertices and edges if the line segment cannot be moved continuously while maintaining its intersections with these vertices and edges.
Analogously, a probe trajectory is isolated by a set of supports if the trajectory cannot be altered without losing any of its supports.
Note that the term ``isolated" has the same meaning as ``extremal," and they will be used interchangeably hereafter.

\begin{lem}
\label{lem1}
Assume that a feasible \textbf{unarticulated} trajectory exists -- that is, the unarticulated probe can be inserted into $S$ to reach $t$ while avoiding the obstacles (i.e., $t$ can see to infinity without obstruction).
Then, there exists an extremal feasible unarticulated trajectory such that the probe is isolated by either one support vertex or two support edges.
\end{lem}

\begin{proof}
\label{lem1_proof}
Let $T$ be a feasible probe trajectory such that its final configuration is unarticulated.
Specifically, in the final configuration of the probe, $abc$ is a straight line segment, and point $c$ coincides with $t$.

Let $\Pi$ be an arbitrary plane passing through $t$ and containing line segment $abc$.
Assume that plane $\Pi$ intersects at least one triangular obstacle.
Observe that such an arbitrary plane $\Pi$ always exists; otherwise, workspace $S$ must be free of obstacles.
Let $T'$ be the trajectory resulting from rotating $T$ around $t$ in plane $\Pi$ until $T$ intersects an obstacle edge at either i) an endpoint or ii) an interior point.

If the intersection occurs at an endpoint $v$ of the obstacle edge, then we have a feasible trajectory $T'$ that is isolated by a support vertex $v$ (for $t$ and $v$ define a unique line).

On the other hand, suppose that $T'$ intersects the obstacle edge at an interior point.
Let $e_1$ denote the intersecting obstacle edge, and $\Psi$ be the unique plane containing $e_1$ and $t$.
When $T'$ is rotated around $t$ in plane $\Psi$ until $T'$ intersects an endpoint of $e_1$ or some obstacle edge $e_2$, a new feasible trajectory $T''$ is obtained.
If $T''$ intersects $e_1$ or $e_2$ at an endpoint, then $T''$ is isolated with respect to a support vertex.
If $T''$ intersects $e_2$ at an interior point, then $T''$ is isolated with respect to two support edges (i.e., $e_1$ and $e_2$).
\end{proof}

\begin{lem}
\label{lem2}
Assume that a feasible \textbf{articulated} trajectory exists -- that is, the unarticulated probe can be inserted into $S$, and a subsequent rotation of $bc$ around $b$ can be performed to reach $t$, all while avoiding the obstacles.
Then, there exists either
\renewcommand{\labelenumi}{\roman{enumi})}
\begin{enumerate}
	\item an extremal feasible unarticulated trajectory or 
	\item an extremal feasible trajectory such that the probe assumes an articulated final configuration, and the trajectory is isolated with respect to at most four supports (edges, vertices, surfaces, or a combination of the three).
\end{enumerate}
\end{lem}

\begin{proof}
\label{lem2_proof}
Let $T$ denote a feasible probe trajectory such that its final configuration is articulated.
Namely, in the final configuration of the probe, $ab$ and $bc$ are not collinear, and point $c$ coincides with $t$.

Let $\Pi$ be the unique plane containing $ab$ and $t$.
Since the trajectory of the probe is planar (i.e., in plane $\Pi$), without loss of generality, assume that $bc$ of $T$ is rotated clockwise around $b$ to reach $t$ in plane $\Pi$.
For ease of subsequent discussion, $bc$ and $bt$ (of an articulated trajectory) would be used to indicate line segment $bc$ of the probe in its intermediate and final configurations, respectively.

A feasible trajectory $T'$ can be obtained by rotating $ab$ of $T$ around $b$ in clockwise direction in plane $\Pi$ until either $ab$ and $bt$ become collinear or $ab$ intersects an obstacle edge $e_1$ outside $C$.
In the former case, $T'$ has an unarticulated final configuration, and by directly applying Lemma \ref{lem1}, we obtain an extremal feasible unarticulated trajectory as a result.
In the latter case, let $T''$ be the trajectory resulting from rotating $bt$ of $T'$ around $t$ counter-clockwise in plane $\Pi$, while keeping $ab$ intersecting $e_1$, until either 1) $bt$ or 2) $ab$ intersects some obstacle edge $e_2$.
Trajectory $T''$ is feasible, and the proof is given in \cite{teo20traj}.

\paragraph{Case 1: $bt$ intersects $e_2$ (inside $C$)}
Let $p_1$ be the intersection point between $ab$ and $e_1$.
Let ${e_2}'$ be the projection of $e_2$ onto the surface of sphere $C$ from center $t$.
Notice that ${e_2}'$ is a circular arc on $C$, and point $b$ of trajectory $T''$ must lie on ${e_2}'$ as long as $bt$ intersects $e_2$ (see Figure \ref{fig_c1}).

\begin{figure}[h]
	\centering
	\includegraphics[scale=0.24]{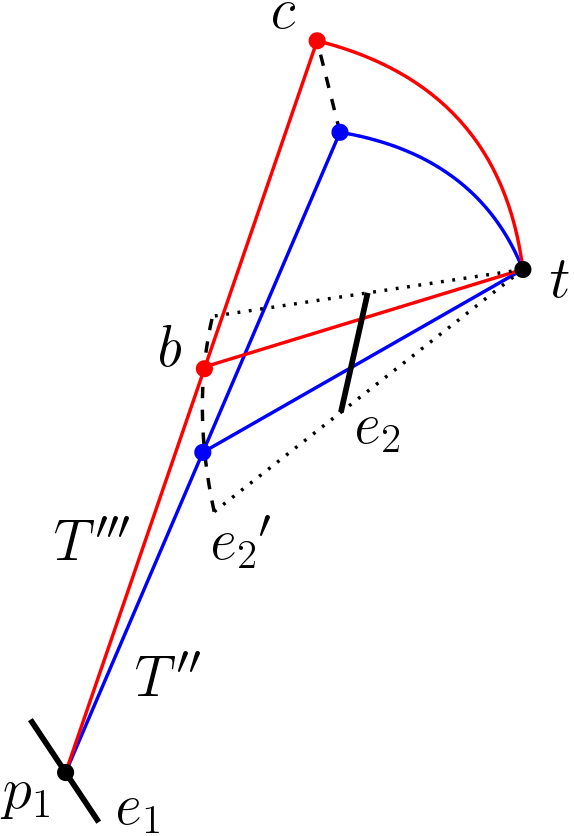}
	\caption{Illustration for Case 1.}
	\label{fig_c1}
\end{figure}

Let $\sigma_{bct}$ denote the minor circular sector (with a central angle $\leq \pi/2$ radians) bounded by radii $bc$ and $bt$, and let $\gamma_{ct}$ be the circular arc of $\sigma_{bct}$.
Let point $b$ of $T''$ be translated along ${e_2}'$ in the direction such that the obtuse angle $\angle abt$ increases, while keeping $ab$ intersecting $e_1$ at $p_1$, until
a) $bt$, b) $bc$, c) $ab$, d) $\sigma_{bct}$, or e) $\gamma_{ct}$ intersects with an obstacle edge $e_3$,
f) $\gamma_{ct}$ is tangent to the surface of some obstacle triangle,
g) $bt$ intersects with an endpoint of $e_2$, or
h) $ab$ becomes collinear with $bt$.
Let $T'''$ denote the resulting trajectory, as depicted in Figure \ref{fig_c1}.
Since $T'''$ is obtained by simply perturbing $T''$ until the trajectory comes into contact with an obstacle (or $\angle abt$ reaches $\pi$ radians), $T'''$ is feasible.

\paragraph{Case 1(a): $bt$ intersects $e_3$}
Note that two generic line segments and a point induce an incident line in 3D space.
Since $t$ is fixed and $bt$ intersects two obstacle edges $e_2$ and $e_3$, $bt$ is isolated (and thus point $b$ is fixed).

Let $\Psi$ be the unique plane containing $e_1$ and $b$.
As illustrated in Figure \ref{fig_c1a-d} Case 1(a),  let $T''''$ be the trajectory obtained from rotating $abc$ of $T'''$ around $b$ in plane $\Psi$ in the direction such that the obtuse angle $\angle abt$ increases until
i) $bc$, ii) $ab$, iii) $\sigma_{bct}$, or iv) $\gamma_{ct}$ intersects with some obstacle edge $e_4$,
v) $\gamma_{ct}$ is tangent to the surface of an obstacle triangle,
vi) $ab$ intersects an endpoint of $e_1$, or 
vii) $ab$ becomes collinear with $bt$.
Obviously, given that $T''''$ is derived from $T'''$ through a straightforward perturbation, $T''''$ is feasible.

\begin{figure}[h]
	\centering
	\includegraphics[scale=0.24]{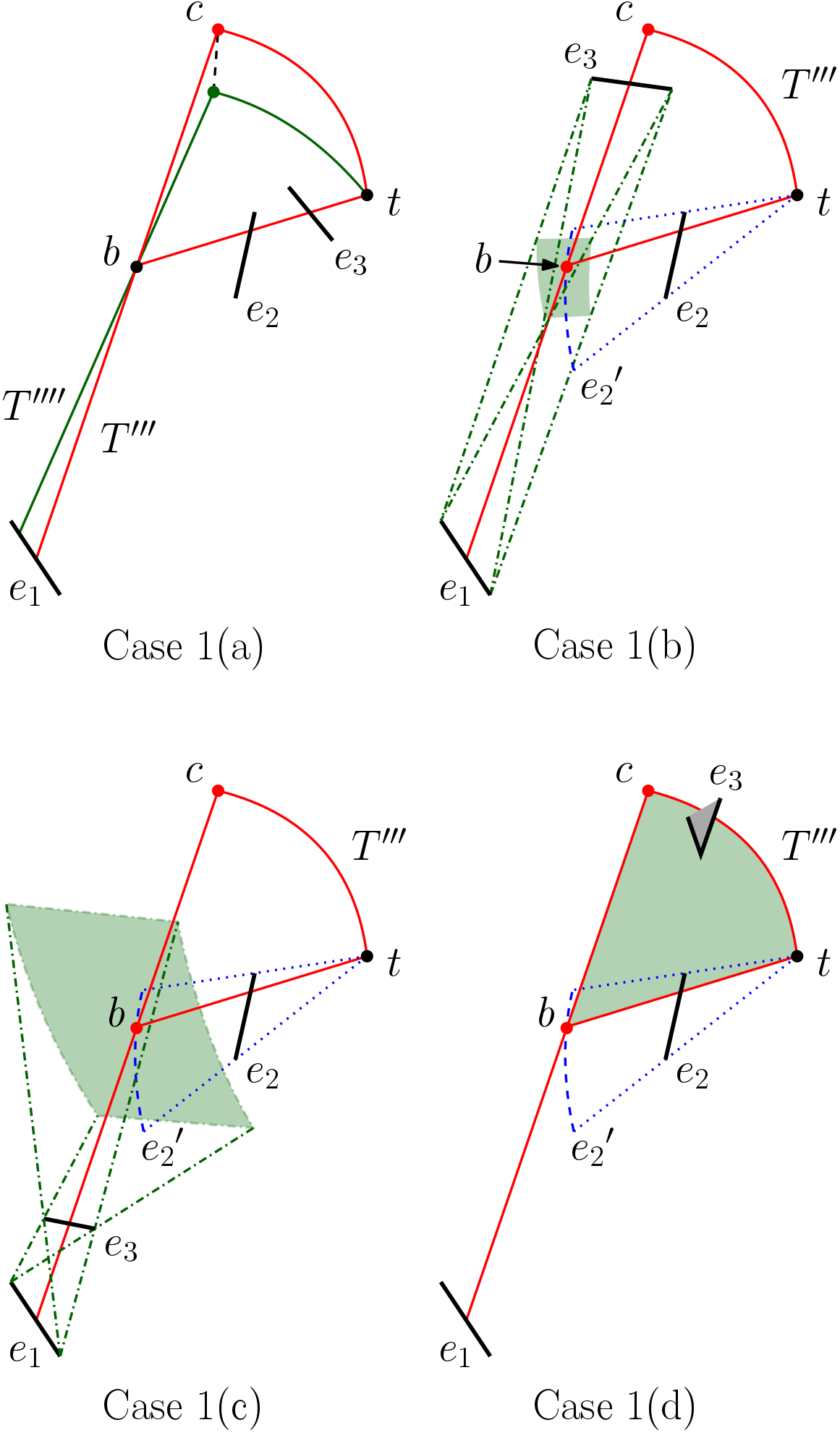}
	\caption{Illustrations for Cases 1(a)--(d).}
	\label{fig_c1a-d}
\end{figure}

\paragraph{Case 1(a)(i), (ii): $bc$ or $ab$ intersects $e_4$}
Since $abc$ intersects two obstacle edges $e_1$ and $e_4$ (and $b$ is fixed), $abc$ is isolated.
Thus, $T''''$ is an extremal feasible articulated trajectory.

\paragraph{Case 1(a)(iii): $\sigma_{bct}$ intersects $e_4$ (at one of its endpoints)}
In 3D space, the intersection of two generic planes is a line.
Since the plane containing $\sigma_{bct}$ and the plane $\Psi$ containing $e_1$ and $b$ are fixed, their intersection line, which supports $abc$, is fixed.
Hence, $T''''$ is an isolated feasible articulated trajectory.

\paragraph{Case 1(a)(iv), (v): $\gamma_{ct}$ intersects $e_4$ or is tangent to the surface of an obstacle triangle}
The plane containing $\sigma_{bct}$ is isolated, and so is $abc$ (for more insight, see Cases 1(e) and (f) below).
As a result, $T''''$ is an isolated feasible articulated trajectory.

\paragraph{Case 1(a)(vi): $ab$ intersects an endpoint of $e_1$}
Since $b$ is fixed, $abc$ is isolated with respect to $b$ and the intersected endpoint of $e_1$.
Thus, $T''''$ is an extremal feasible articulated trajectory.

\paragraph{Case 1(a)(vii): $ab$ becomes collinear with $bt$}
In this case, by directly applying Lemma \ref{lem1}, we claim that $T''''$ is an extremal feasible unarticulated trajectory.

\paragraph{Case 1(b): $bc$ intersects $e_3$}
Note that $ab$ intersects $e_1$, and $bc$ intersects $e_3$.
Let $P(e_1, e_3)$ be the set of points $q$ for which there exists a line segment starting at $e_1$, passing through $q$, and ending at $e_3$.
Let $I$ be the intersection of $P(e_1, e_3)$, $C$, and ${e_2}'$ (for an illustration, see Figure \ref{fig_c1a-d} Case 1(b), where $I$ is indicated as the portion of the blue dashed line lying in the green shaded area).
Notice that $I$ is a circular arc on $C$.
An isolated feasible trajectory can be obtained by translating point $b$ of $T'''$ along $I$ in the direction such that the obtuse angle $\angle abt$ increases until
i) $bt$, $bc$, $ab$, $\sigma_{bct}$, or $\gamma_{ct}$ intersects with some obstacle edge $e_4$,
ii) $\gamma_{ct}$ is tangent to the surface of an obstacle triangle,
iii) $bt$, $bc$, or $ab$ reaches an endpoint of its currently intersecting obstacle edge, or
iv) $ab$ becomes collinear with $bt$.
The detailed analyses for Cases 1(b)(i)--(iv) (as well as similar others in the cases that follow) are omitted herein due to their similarity to those in Case 1(a).

\paragraph{Case 1(c): $ab$ intersects $e_3$}
Notice that $ab$ intersects two edges $e_1$ and $e_3$.
Assume, without loss of generality, that $e_1$ and $e_3$ are intersected by $ab$ in the way depicted in Figure \ref{fig_c1a-d} Case 1(c).
Let $P(e_1, e_3)$ be the set of points $q$ for which there exists a ray starting at $e_1$, passing through $e_3$, and then through $q$.
Let $I$ be the intersection of $P(e_1, e_3)$, $C$, and ${e_2}'$.
Note that $I$ is a circular arc on $C$.
Circular arc $I$ is illustrated, in Figure \ref{fig_c1a-d} Case 1(c), as the portion of the blue dashed line lying in the green shaded area.
An isolated feasible trajectory can be obtained by moving point $b$ of $T'''$ along $I$ in the direction that increases the obtuse angle $\angle abt$ until
i) $bt$, $bc$, $ab$, $\sigma_{bct}$, or $\gamma_{ct}$ intersects with an obstacle edge $e_4$,
ii) $\gamma_{ct}$ is tangent to the surface of an obstacle triangle,
iii) $bt$ or $ab$ reaches an endpoint of its currently intersecting obstacle edge, or
iv) $ab$ becomes collinear with $bt$.

\paragraph{Case 1(d): $\sigma_{bct}$ intersects $e_3$ (at one of its endpoints)}
Observe, as shown in Figure \ref{fig_c1a-d} Case 1(d), that $T'''$ can be made isolated if point $b$ of $T'''$ is further translated along ${e_2}'$, while keeping $ab$ intersecting $e_1$ and maintaining the incidence between $\sigma_{bct}$ and $e_3$, until
i) $bt$, $bc$, $ab$, $\sigma_{bct}$, or $\gamma_{ct}$ intersects with some obstacle edge $e_4$,
ii) $\gamma_{ct}$ is tangent to the surface of an obstacle triangle,
iii) $bt$ or $ab$ reaches an endpoint of its currently intersecting obstacle edge, or
iv) $ab$ becomes collinear with $bt$.
The resulting trajectory is feasible.

\paragraph{Case 1(e): $\gamma_{ct}$ intersects $e_3$}
Let $D$ be a sphere of radius $r$ with center $b$ located on sphere $C$.
Note that $D$ passes through $t$.
Consider the scenario that $D$ is tangent to obstacle edge $e_3$.
When $D$ is rotated around $t$ while maintaining its tangency to $e_3$, the center $b$ of $D$ translates along a curve ${e_3}'$ on $C$ (see Figure \ref{fig_c1e}).

\begin{figure}[h]
	\centering
	\includegraphics[scale=0.24]{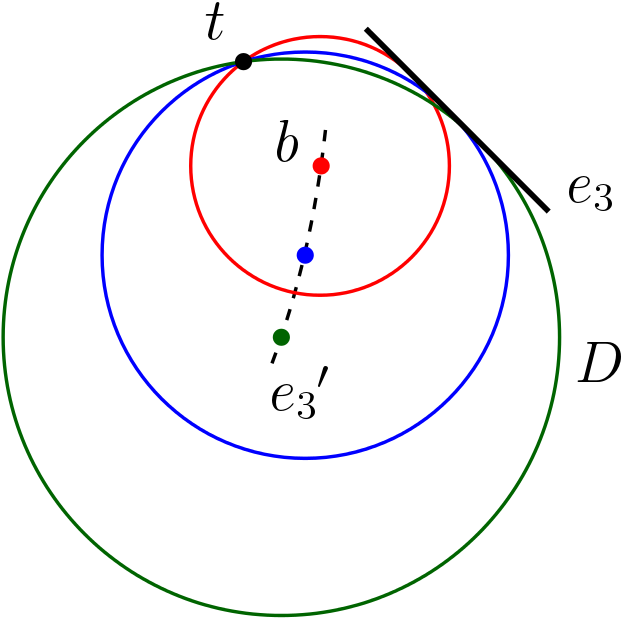}
	\caption{Illustration for Case 1(e).
		A top sectional view along the plane containing $e_3$ and $t$.}
	\label{fig_c1e}
\end{figure}

A plane can be defined by using two points on $D$ and the center $b$ of $D$.
Consider a sphere $D$ with center $b$ on ${e_3}'$.
Sphere $D$ intersects $t$ and is tangent to $e_3$ at a point $p$.
Thus, we can define a plane $\Pi_b$ containing $t$, $p$, and $b$, for any given point $b$ on ${e_3}'$.

Notice that $\gamma_{ct}$ is a circular arc on $D$ originating at $t$, and $\gamma_{ct}$ must be on the circle of intersection between $D$ and $\Pi_b$, for some point $b$ on ${e_3}'$.
Since ${e_2}'$ and ${e_3}'$ intersect at a point (i.e., point $b$ on $C$), point $b$ is fixed.
As a result, plane $\Pi_b$ is set, and so is the circle on which $\gamma_{ct}$ lies.
If $\Pi_b$ does not contain $e_1$, then $abc$ is isolated.
Otherwise, $abc$ can be rotated around $b$, while continuing to intersect $e_1$, until either $ab$ or $bc$ intersects an obstacle edge (or an endpoint of $e_1$).
In the end of either scenario, an extremal feasible trajectory is obtained.

\paragraph{Case 1(f): $\gamma_{ct}$ is tangent to the surface of an obstacle triangle}
The analysis is similar to Case 1(e).
Let $\tau$ denote the obstacle triangle to whose surface $\gamma_{ct}$ is tangent.
As $D$ is rotated around $t$ while remaining tangent to the surface of $\tau$, the center $b$ of $D$ moves along a curve $\kappa_\tau$ on $C$.
Given that ${e_2}'$ and $\kappa_\tau$ intersect at a point, $b$ is fixed.
The same argument as in Case 1(e) immediately follows.

\paragraph{Case 1(g): $bt$ intersects an endpoint of $e_2$}
Given that $t$ is fixed, $bt$ of $T'''$ is isolated by its intersection with an endpoint of $e_2$.
A similar analysis as in Case 1(a) then follows.

\paragraph{Case 1(h): $bc$ becomes collinear with $bt$}
In this case, Lemma \ref{lem1} is directly applicable, resulting in an extremal feasible unarticulated trajectory.

\paragraph{Case 2: $ab$ intersects $e_2$ (outside $C$)}
Using similar arguments as given in Case 1, we can obtain an isolated feasible articulated trajectory such that i) $ab$ intersects two edges, and ii) any of $bt$, $bc$, $ab$, $\sigma_{bct}$, and $\gamma_{ct}$ are supported by at most two vertices, edges, and/or surfaces. \\

One may find some overlaps between Cases 1 and 2.
On the whole, an extremal feasible articulated trajectory is characterized by its intersections (or tangencies) with at most four supports consisting of obstacle endpoints, edges, surfaces, or a combination of the three.
This concludes the proof of Lemma \ref{lem2}.
\end{proof}

\section{Finding and validating extremal trajectories}
\label{comp_valid}
Based on Lemmas \ref{lem1} and \ref{lem2}, we can compute the set of extremal probe trajectories and verify them for feasibility.
The algorithms and data structures required are discussed in this section.

\subsection{Extremal feasible unarticulated trajectories}
Let $V$ and $E$ denote the set of vertices and edges of the triangles of $P$, respectively.
We compute the set $R$ of rays, each of which i) originates at point $t$, ii) is either passing through a vertex of $V$ or isolated with respect to two support edges of $E$, and iii) does not intersect any triangle of $P$ in its interior.
Each ray of $R$ represents an extremal feasible unarticulated trajectory.
Set $R$ can be obtained by computing the visibility polyhedron from point $t$ in $O(n^2)$ time using $O(n^2)$ space \cite{mckenna87worst}.
In fact, this approach yields all feasible unarticulated solutions, since the visibility polyhedron gives all unbounded rays from $t$ to infinity (i.e., feasible unarticulated trajectories).

\begin{lem}
\label{lem3}
The set of all feasible unarticulated probe trajectories can be found in $O(n^2)$ time using $O(n^2)$ space.
\end{lem}

\subsection{Extremal feasible articulated trajectories}
We begin by computing the set of extremal articulated trajectories, which are characterized by $O(n^4)$ combinatorial events (see Lemma \ref{lem2}).
As detailed next, these extremal articulated trajectories can be enumerated in $O(n^4)$ time either geometrically or algebraically.

\subsubsection{Computing extremal articulated trajectories}
\label{comp}
Table \ref{tab1} lists all the distinct cases of isolated articulated trajectories, each of which is indicated by the number of obstacle edges (i.e., support edges) intersected by line segments $ab$, $bc$, $bt$, and circular arc $\gamma_{ct}$, and the number of obstacle endpoints (i.e., support vertices) incident on the interior of circular sector $\sigma_{bct}$.
Table \ref{tab1} also provides, for each case, a brief description of the method (i.e., a sequence of geometric operations) and the worst-case time complexity for finding an extremal articulated trajectory.
When describing the geometric operations in each case, obstacles edges are divided into those lying inside and outside $C$.
Given that an obstacle edge may intersect $C$ at most twice, an obstacle edge may be partitioned into at most three line segments.

\small
\setlength{\tabcolsep}{5pt}
\renewcommand{\arraystretch}{1.5}
\begin{longtable}[c]{|c|c|c|c|c|c|P{6cm}|c|}
	\caption{Extremal articulated trajectories.}
	\label{tab1} \\
	
	\hline
	\textbf{Case} & \boldmath$ab$ & \boldmath$bc$ & \boldmath$bt$ & \boldmath$\sigma_{bct}$ & \boldmath$\gamma_{ct}$ & \textbf{Method} & \textbf{Time} \\
	\hline
	\endfirsthead
	
	\multicolumn{8}{|c|}{Continued on next page.} \\
	\hline
	\endfoot
	
	\hline
	\textbf{Case} & \boldmath$ab$ & \boldmath$bc$ & \boldmath$bt$ & \boldmath$\sigma_{bct}$ & \boldmath$\gamma_{ct}$ & \textbf{Method} & \textbf{Time} \\
	\hline
	\endhead
	
	\endlastfoot
	
	1* & 1 & 1 & 1 & 1 &   & See main text in Section \ref{comp}. & $n^4$ \\
	\hline
	
	2 & 1 & 2 & 1 &   &   &
	\begin{itemize}[label={--},noitemsep,leftmargin=*,topsep=0pt,partopsep=0pt]
		\item For any three obstacle edges $e_1$, $e_2$, and $e_3$ (one outside $C$, and two inside or outside $C$), compute the projection (curve $\kappa_1$) of the hyperboloid that contains $e_1$, $e_2$, and $e_3$ on $C$.
		\item For every obstacle edge $e_4$ inside $C$, compute its projection (curve $\kappa_2$) on $C$.
		\item Compute the intersection of $\kappa_1$ and $\kappa_2$.
		\vspace*{-\baselineskip}
	\end{itemize}
	& $n^4$ \\
	\hline
	
	3 & 1 & 1 & 2 &   &   &
	\begin{itemize}[label={--},noitemsep,leftmargin=*,topsep=0pt,partopsep=0pt]
		\item For any two obstacle edges $e_1$ and $e_2$ (one inside $C$, and the other inside or outside $C$), compute the projection (area $A$) of $P(e_1, e_2)$ on $C$.
		\item For any two obstacle edges $e_3$ and $e_4$ inside $C$, compute the radial line segment $s$ originating at $t$ and intersecting $e_3$ and $e_4$, and determine the intersection point $b$ between $s$ and $C$.
		\item Note that a distinct $abc$ is given by a point $b$ lying on $A$.
		\vspace*{-\baselineskip}
	\end{itemize}
	& $n^4$ \\
	\hline
	
	4 & 1 &   & 1 & 2 &   &
	\begin{itemize}[label={--},noitemsep,leftmargin=*,topsep=0pt,partopsep=0pt]
		\item For any two obstacle endpoints $v_1$ and $v_2$ (within distance $2r$ from $t$), compute the plane $\Pi$ that contains $v_1$, $v_2$, and $t$.
		\item For any obstacle edge $e_1$ outside $C$, compute the intersection of $e_1$ and $\Pi$.
		\item For any obstacle edge $e_2$ inside $C$, compute the intersection of $e_2$ and $\Pi$.
		\vspace*{-\baselineskip}
	\end{itemize}
	& $n^4$ \\
	\hline
	
	5 & 1 &   & 2 & 1 &   &
	\begin{itemize}[label={--},noitemsep,leftmargin=*,topsep=0pt,partopsep=0pt]
		\item For any two obstacle edges $e_1$ and $e_2$ inside $C$, compute the radial line segment $s$ originating at $t$ and intersecting $e_1$ and $e_2$, and determine the intersection point $b$ between $s$ and $C$.
		\item For any obstacle endpoint $v$ (within distance $2r$ from $t$), compute the plane $\Pi$ that contains $v$, $b$, and $t$.
		\item For any obstacle edge $e_3$ outside $C$, compute the intersection of $e_3$ and $\Pi$.
		\vspace*{-\baselineskip}
	\end{itemize}
	& $n^4$ \\
	\hline
	
	6* & 2 &   & 1 & 1 &   & See main text in Section \ref{comp}. & $n^4$ \\
	\hline
	
	7* & 2 & 1 &   & 1 &   & See main text in Section \ref{comp}. & $n^4$ \\
	\hline
	
	8 & 2 & 1 & 1 &   &   &
	\begin{itemize}[label={--},noitemsep,leftmargin=*,topsep=0pt,partopsep=0pt]
		\item For any three obstacle edges $e_1$, $e_2$, and $e_3$ (two outside $C$, and one inside or outside $C$), compute the projection (curve $\kappa_1$) of the hyperboloid that contains $e_1$, $e_2$, and $e_3$ on $C$.
		\item For any obstacle edge $e_4$ inside $C$, compute its projection (curve $\kappa_2$) on $C$.
		\item Compute the intersection of $\kappa_1$ and $\kappa_2$.
		\vspace*{-\baselineskip}
	\end{itemize}
	& $n^4$ \\
	\hline
	
	9 & 2 & 2 &   &   &   &
	For any four obstacle edges (two outside $C$, and two inside or outside $C$), compute the incidence lines induced by the four obstacle edges.
	& $n^4$ \\
	\hline
	
	10 & 2 &   & 2 &   &   &
	\begin{itemize}[label={--},noitemsep,leftmargin=*,topsep=0pt,partopsep=0pt]
		\item For any two obstacle edges $e_1$ and $e_2$ outside $C$, compute the projection (area $A$) of $P(e_1, e_2)$ on $C$.
		\item For any two obstacle edges $e_3$ and $e_4$ inside $C$, compute the radial line segment $s$ originating at $t$ and intersecting $e_3$ and $e_4$, and determine the intersection point $b$ between $s$ and $C$.
		\item Note that a distinct $abc$ is given by a point $b$ lying on $A$.
		\vspace*{-\baselineskip}
	\end{itemize}
	& $n^4$ \\
	\hline
	
	11 & 2 &   &   & 2 &   &
	\begin{itemize}[label={--},noitemsep,leftmargin=*,topsep=0pt,partopsep=0pt]
		\item For any two obstacle endpoints $v_1$ and $v_2$ (within distance $2r$ from $t$), compute the plane $\Pi$ that contains $v_1$, $v_2$, and $t$.
		\item For any two obstacle edges $e_1$ and $e_2$ outside $C$, compute their intersections with $\Pi$.
		\vspace*{-\baselineskip}
	\end{itemize}
	& $n^4$ \\
	\hline
	
	12 & 3 & 1 &   &   &   &
	For any four obstacle edges (three outside $C$, and one inside or outside $C$), compute the incidence lines induced by the four obstacle edges.
	& $n^4$ \\
	\hline
	
	13 & 3 &   & 1 &   &   &
	\begin{itemize}[label={--},noitemsep,leftmargin=*,topsep=0pt,partopsep=0pt]
		\item For any three obstacle edges $e_1$, $e_2$, and $e_3$ outside $C$, compute the projection (curve $\kappa_1$) of the hyperboloid that contains $e_1$, $e_2$, and $e_3$ on $C$.
		\item For any obstacle edge $e_4$ inside $C$, compute its projection (curve $\kappa_2$) on $C$.
		\item Compute the intersection of $\kappa_1$ and $\kappa_2$.
		\vspace*{-\baselineskip}
	\end{itemize}
	& $n^4$ \\
	\hline
	
	14* & 3 &   &   & 1 &   & See main text in Section \ref{comp}. & $n^4$ \\
	\hline
	
	15 & 4 &   &   &   &   &
	For any four obstacle edges outside $C$, compute the incidence lines induced by the four obstacle edges.
	& $n^4$ \\
	\hline
	
	16 & 1 &   & 1 &   & 1 &
	\begin{itemize}[label={--},noitemsep,leftmargin=*,topsep=0pt,partopsep=0pt]
		\item For any obstacle edge $e_1$ (within distance $2r$ from $t$), compute the curve $\kappa_1$ on which point $b$ (i.e., the center of the sphere $D$ tangent to $e_1$) lies.
		\item For any obstacle edge $e_2$ inside $C$, compute its projection (curve $\kappa_2$) on $C$.
		\item Compute the intersection point $b$ between $\kappa_1$ and $\kappa_2$.
		\item Compute the plane $\Pi$ that contains $t$, $b$, and the point of intersection/tangency between $e_1$ and $D$.
		\item For any obstacle edge $e_3$ outside $C$, compute the intersection of $\Pi$ and $e_3$.
		\item Assume that $\Pi$ is not the same as the plane that contains $b$ and $e_3$.
		\vspace*{-\baselineskip}
	\end{itemize}
	& $n^3$ \\
	\hline
	
	17 & 1 & 1 & 1 &   & 1 &
	\begin{itemize}[label={--},noitemsep,leftmargin=*,topsep=0pt,partopsep=0pt]
		\item This is similar to Case 16 above with a different assumption at the end.
		\item Assume that $\Pi$ is the same as the plane that contains $b$ and $e_3$.
		\item For any obstacle edge $e_4$ inside or outside $C$, compute the intersection of $\Pi$ and $e_4$.
		\vspace*{-\baselineskip}
	\end{itemize}
	& $n^4$ \\
	\hline
	
	18 & 1 &   & 1 & 1 & 1 &
	This can be thought of as Case 16 with an additional obstacle vertex supporting $\sigma_{bct}$ (which does not contribute to additional extremal trajectories).
	& $n^3$ \\
	\hline
	
	19 & 1 &   & 2 &   & 1 &
	This can be thought of as Case 16 with an additional obstacle edge intersecting $bt$ (which does not contribute to additional extremal trajectories).
	& $n^3$ \\
	\hline
	
	20 & 2 & 1 &   &   & 1 &
	\begin{itemize}[label={--},noitemsep,leftmargin=*,topsep=0pt,partopsep=0pt]
		\item For any three obstacle edges $e_1$, $e_2$, and $e_3$ (two outside $C$ and one inside or outside $C$), compute the projection (curve $\kappa_1$) of the hyperboloid that contains $e_1$, $e_2$, and $e_3$ on $C$.
		\item For any obstacle edge $e_4$ (within distance $2r$ from $t$), compute the curve $\kappa_2$ on which point $b$ (i.e., the center of the sphere $D$ tangent to $e_4$) lies.
		\item Compute the intersection point $b$ between $\kappa_1$ and $\kappa_2$.
		\vspace*{-\baselineskip}
	\end{itemize}
	& $n^4$ \\
	\hline
	
	21 & 2 &   & 1 &   & 1 &
	\begin{itemize}[label={--},noitemsep,leftmargin=*,topsep=0pt,partopsep=0pt]
		\item This is similar to Case 16 with a different assumption at the end.
		\item Assume that $\Pi$ is the same as the plane that contains $b$ and $e_3$.
		\item For any obstacle edge $e_4$ outside $C$, compute the intersection of $\Pi$ and $e_4$.
		\vspace*{-\baselineskip}
	\end{itemize}
	& $n^4$ \\
	\hline
	
	22* & 2 &   &   & 1 & 1 & See main text in Section \ref{comp}. & $n^4$ \\
	\hline
\end{longtable}
\normalsize

For conciseness, certain scenarios are omitted from Table \ref{tab1} given their trivial nature, and they include those involving the isolation of a feasible articulated trajectory due to incidence of its line segment (i.e., $ab$, $bc$, or $bt$) with obstacle endpoints (i.e., support vertices).
In addition, the cases in which $\gamma_{ct}$ is tangent to the surface of an obstacle triangle are omitted due to their similarity to those where $\gamma_{ct}$ is tangent to an obstacle edge.

Note that the set of extremal articulated trajectories, in all cases besides those denoted by * in Table \ref{tab1}, can be found by using a sequence of simple geometric operations (e.g., computing the intersection of two line segments or planes).
An extremal articulated trajectory in each of the five cases designated by * can be computed by using an algebraic-geometric approach.
Given that these five cases are characteristically similar, as a demonstration, let us take a closer look at Case 1 from an algebraic point of view.

Recall that, in Case 1, $ab$ intersects an edge $e_1$, $bt$ intersects an edge $e_2$, and $bc$ intersects an edge $e_3$.
In addition, $\sigma_{bct}$ intersects an obstacle endpoint $v$.
Let $p$ (resp. $q$) denote the intersection between $ab$ and $e_1$ (resp. $bt$ and $e_2$).

Edges $e_1$, $e_2$, and $e_3$ can be expressed in parametric form as follows:
\begin{gather*}
e_1 = \{ (1 - \lambda_1) u_1 + \lambda_1 v_1 : \lambda_1 \in \mathbb{R}, 0 \leq \lambda_1 \leq 1 \} \\
e_2 = \{ (1 - \lambda_2) u_2 + \lambda_1 v_2 : \lambda_2 \in \mathbb{R}, 0 \leq \lambda_2 \leq 1 \} \\
e_3 = \{ (1 - \lambda_3) u_3 + \lambda_1 v_3 : \lambda_3 \in \mathbb{R}, 0 \leq \lambda_3 \leq 1 \}
\end{gather*}
where $u_i, v_i \in \mathbb{R}^3$ are the two endpoints of edge $e_i$ for $i = 1, 2, 3$.
Analogously, line segments $bt$ and $pq$ can be defined, respectively, as
\begin{gather*}
bt = \{ (1 - \lambda_{bt}) b + \lambda_{bt} t : \lambda_{bt} \in \mathbb{R}, 0 \leq \lambda_{bt} \leq 1 \} \\
pq = \{ (1 - \lambda_{pq}) p + \lambda_{pq} q : \lambda_{pq} \in \mathbb{R}, 0 \leq \lambda_{pq} \leq 1 \}
\end{gather*}
where $b, p, q, t \in \mathbb{R}^3$.
Since $bt$ intersects $e_3$,
\begin{equation}
\label{eq1}
(1 - \lambda_{bt}) b + \lambda_{bt} t = (1 - \lambda_3) u_3 + \lambda_1 v_3
\end{equation}
Furthermore, given that $b$ must lie on a sphere of fixed radius $r$ centered at $t$,
\begin{equation}
\label{eq2}
(b - t) \cdot (b - t) = r^2
\end{equation}
where $\cdot$ designates the dot product.
Note that $b$ must also be located on $pq$.
Thus,
\begin{equation}
\label{eq3}
b = (1 - \lambda_{pq}) p + \lambda_{pq} q
\end{equation}
Since $p$ and $q$ are points on $e_1$ and $e_2$, respectively,
\begin{gather}
\label{eq4}
p = (1 - \lambda_1) u_1 + \lambda_1 v_1 \\
\label{eq5}
q = (1 - \lambda_2) u_2 + \lambda_1 v_2
\end{gather}
The normal vector $n$ of the plane $\Pi$ defined by points $t$, $b$, and $v$ can be expressed as
\[ n = (b - t) \times (v - t) \]
where $\times$ denotes the cross product.
Since $p$ (and thus $q$) must be in plane $\Pi$,
\begin{equation}
\label{eq6}
\begin{split}
n \cdot (p - t) &= 0 \\ 
\big[ (b - t) \times (v - t) \big] \cdot (p - t) &= 0
\end{split}
\end{equation}
At last, without loss of generality, $t$ is assumed to be located at the origin:
\begin{equation}
\label{eq7}
t = (0, 0, 0)
\end{equation}
Note that Equations \ref{eq1}, \ref{eq3}, \ref{eq4}, \ref{eq5}, and \ref{eq7} are in vector form, and each of these equations yields three scalar equations (i.e., $x$, $y$, and $z$ components).
For instance, Equation \ref{eq1} can be expressed, in scalar form, as
\begin{gather*}
(1 - \lambda_{bt}) b_x + \lambda_{bt} t_x = (1 - \lambda_3) u_{3x} + \lambda_1 v_{3x} \\
(1 - \lambda_{bt}) b_y + \lambda_{bt} t_y = (1 - \lambda_3) u_{3y} + \lambda_1 v_{3y} \\
(1 - \lambda_{bt}) b_z + \lambda_{bt} t_z = (1 - \lambda_3) u_{3z} + \lambda_1 v_{3z}
\end{gather*}
Similarly, each of the vector variables $b$, $t$, $p$, and $q$ consists of three scalar components (e.g., $b = (b_x, b_y, b_z)$).
Hence, we have, in total, a system of 17 (scalar) parametric equations with 17 unknown (scalar) variables, which can be solved 
in constant time.

All in all, the worst-case running time for computing an extremal articulated trajectory is $O(n^4)$.

\subsubsection{Validating extremal articulated trajectories}
An extremal articulated trajectory is deemed feasible if and only if i) line segment $ab$ and ii) circular sector $\sigma$ do not intersect with any triangular obstacle in its interior.
Checking for these scenarios can be reduced to the following two query problems -- i) ray shooting query and ii) circular sector emptiness query.

\subsubsection*{Ray shooting queries}
A ray shooting query among $n$ interior disjoint triangles can be performed to determine whether a query line segment $ab$ intersects with any triangular obstacle in its interior (i.e., whether a query line segment stabs through the interior of an obstacle triangle).
According to de Berg et al. \cite{deberg94efficient} and Pellegrini \cite{pellegrini93ray}, such a query can be answered in $O(\log n)$ time using $O(n^{4+\epsilon})$ preprocessing time and space, for any constant $ \epsilon > 0$.
Alternatively, by using the data structure proposed by Agarwal and Matousek \cite{agarwal94range}, which provides a trade-off between space and time, a ray shooting query (amidst $n$ triangles) can be answered in $O(n^{1+\epsilon} / m^{1/4})$ time using $O(m)$ storage and $O(m^{1+\epsilon})$ preprocessing time, for any $n \leq m \leq n^4$.
By employing this fairly complex data structure, given that we have $O(n^4)$ queries, when $m = n$, we obtain $O(n^{3/4 + \epsilon})$ time per query and $O(n^{19/4 + \epsilon})$ total time.


\subsubsection*{Circular sector emptiness queries}
In this section, we address a special instance of the circular sector emptiness query problem in 3D. \\

\noindent\textit{Given a set $P$ of $n$ triangles in $\mathbb{R}^3$, preprocess it so that, for a query circular sector $\sigma$ with a fixed radius $r$ and an endpoint of its arc located at fixed point $t$, one can quickly determine whether $\sigma$ intersects $P$.} \\

Let $\Pi$ be a plane passing through point $t$.
We can parameterize $\Pi$ by using two variables $(I, \Omega)$ (see Figure \ref{fig_plane}, where $\Pi_0$ is the plane with $I = 0$ and $\Omega = 0$).

\begin{figure}[h]
	\centering
	\includegraphics[scale=0.19]{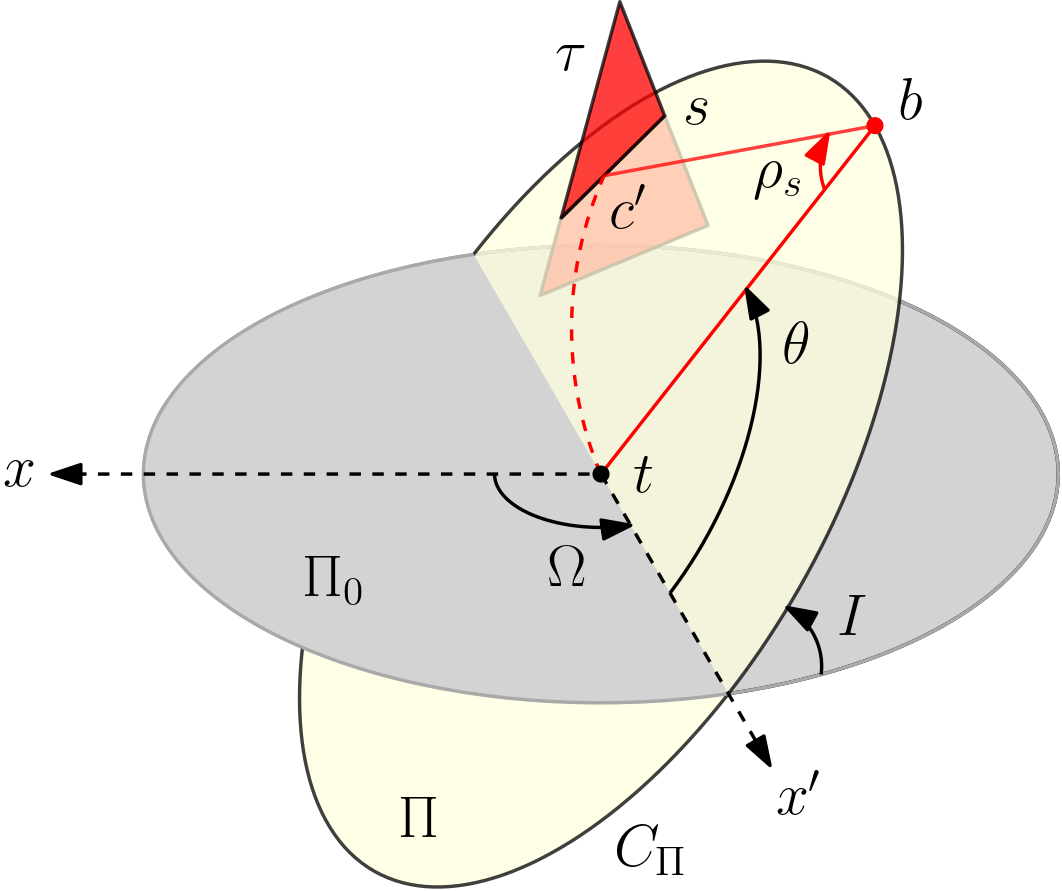}
	\caption{Characterization of $\rho_s$ as function of $\theta$ in $\Pi$.}
	\label{fig_plane}
\end{figure}

Let ${C_\Pi}'$ be the circle in $\Pi$ with radius $\sqrt{2}r$ and the center at $t$.
Let $\tau$ be a triangle that intersects $\Pi$ in a line segment $s$.
Observe that a query circular sector $\sigma$ in $\Pi$, as described above, can only intersect with $\tau$ if $s$ intersects the interior of
${C_\Pi}'$.

Let $C_\Pi$ denote the circle of radius $r$ centered at $t$ in $\Pi$.
For any point $b \in C_\Pi$, let $\theta \in [0, 2\pi)$ be the angle of $tb$ relative to the $x$-axis of $\Pi$ (labeled as the $x'$-axis in Figure \ref{fig_plane}).
Let $bc'$ denote the farthest radius from $bt$, with respect to the circle of radius $r$ centered at $b$ in $\Pi$, before the minor circular sector bounded by radii $bt$ and $bc'$ intersects with $s$.
Let $\rho_s$ be the angle of $bc'$ with respect to $bt$.
Recall that line segment $bc$ of the probe, after the insertion of the probe has completed, may only rotate up to $\pi/2$ radians in either direction in plane $\Pi$.
Thus, $\rho_s \in [0, \pi/2]$.
Here we consider $\rho_s$ to be the clockwise angle from $bt$ in plane $\Pi$.
The other case can be handled symmetrically.

We fix plane $\Pi$ (i.e., $I$ and $\Omega$) and proceed to characterize $\rho_s$ as a function of $\theta$.
Two different cases are examined separately, depending on whether line segment $s$ lies 1) inside $C_\Pi$ or 2) outside $C_\Pi$ and inside ${C_\Pi}'$.
Note that a line segment $s$ may intersect $C_\Pi$ or ${C_\Pi}'$ (or both).
Since a line segment may intersect a circle at most two times, such a line segment can be partitioned into at most three line segments -- one inside and two outside the circle -- before proceeding with the following analysis.

\paragraph{Case 1: $s$ lies inside $C_\Pi$}
For brevity, the quarter-circular sector associated with a point $b$ (i.e., the maximum possible area swept by $bc$ to reach $t$), where the angle of $tb$ relative to the $x'$-axis is $\theta$, is henceforth referred to as the \emph{quart-sector of $\theta$}.

Let $\phi_{s,1}$, $\phi_{s,2}$, and $\phi_{s,3}$ be defined as follows (see Figure \ref{fig_in_a-d}A).
Angle $\phi_{s,1}$ is the smallest angle $\theta$ at which the circular arc of the quart-sector of $\theta$ intersects with $s$ at one of its endpoints or an interior point.
Angle $\phi_{s,2}$ is the smallest angle $\theta$ at which the radius $bt$ of the quart-sector of $\theta$ intersects with $s$ at one of its endpoints.
Angle $\phi_{s,3}$ is the largest angle $\theta$ at which the radius $bt$ of the quart-sector of $\theta$ intersects with $s$ at one of its endpoints.
In other words, as $\theta$ varies from 0 to $2\pi$, $\phi_{s,1}$ and $\phi_{s,3}$ are the angles $\theta$ at which the quart-sector of $\theta$ first and last intersects with $s$, respectively.

\begin{figure}[h]
	\centering
	\includegraphics[scale=0.18]{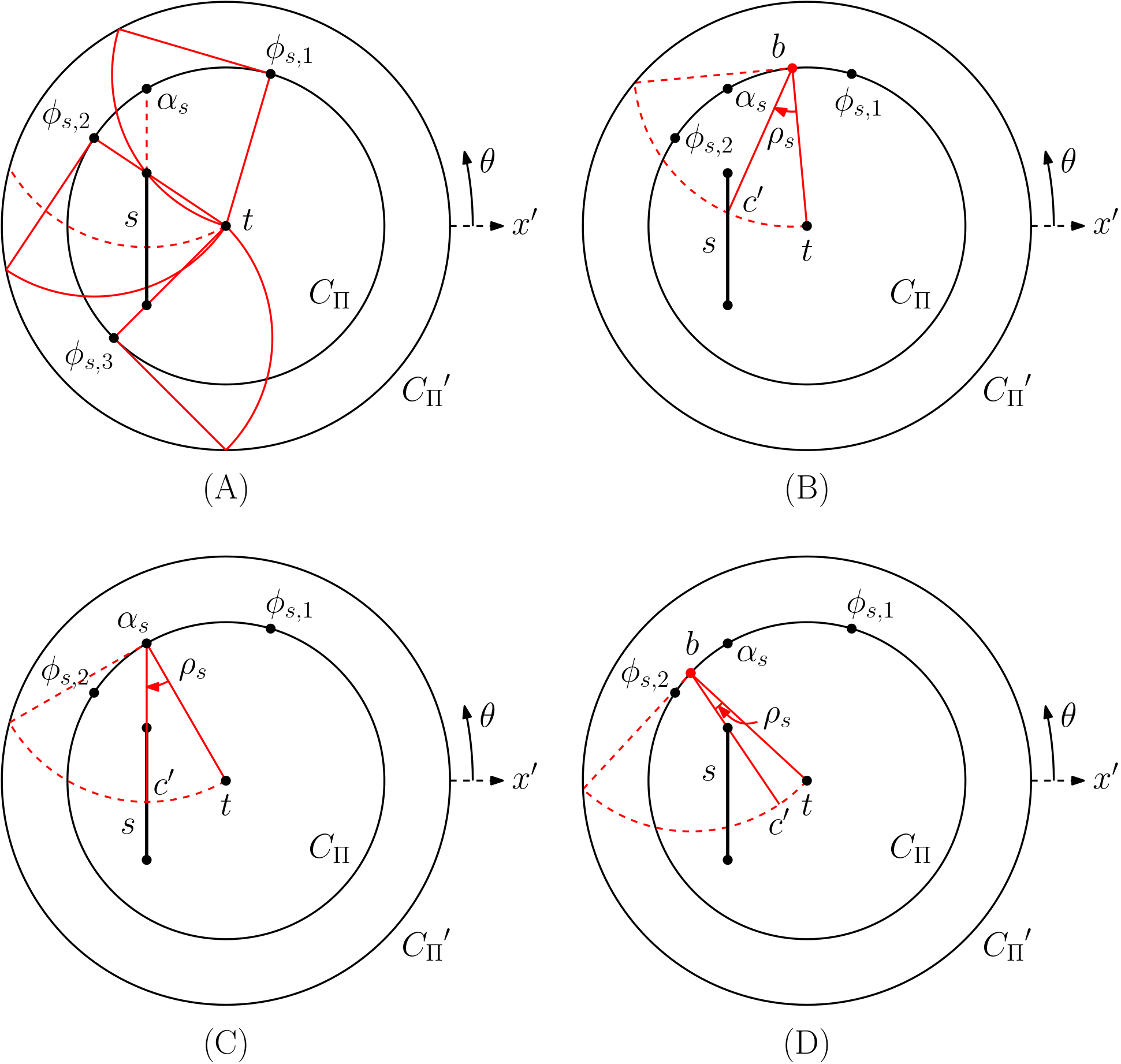}
	\caption{(A) Characterizing $\rho_s(\theta)$ for $\theta \in [\phi_{s,1}, \phi_{s,2}]$ in Case 1.
		Illustrations for $\rho_s(\theta)$ when (B) $\phi_{s,1} < \theta < \alpha_s$, (C) $\theta = \alpha_s$, and (D) $\alpha_s < \theta < \phi_{s,2}$, respectively.}
	\label{fig_in_a-d}
\end{figure}

For a line segment $s$ lying inside $C_\Pi$, as shown in Figure \ref{fig_in_a-d}, we are only concerned with computing $\rho_s$ for $\theta \in [\phi_{s,1}, \phi_{s,2}]$, given that $\theta \in [\phi_{s,2}, \phi_{s,3}]$ is infeasible due to intersection of $bt$ with $s$, and $\rho_s = \pi/2$ for $\theta \in [0, \phi_{s,1}] \cup [\phi_{s,3}, 2\pi)$.

Let $\alpha_s$ be the angle $\theta$ of the intersection point between $C_\Pi$ and the supporting line of $s$.
For $\theta \in [\phi_{s,1}, \phi_{s,2}]$, $\rho_s(\theta)$ can be represented by a piecewise continuous curve, which consists of at most two pieces, corresponding to two sub-intervals $[\phi_{s,1}, \alpha_s]$ and $[\alpha_s, \phi_{s,2}]$.
Note that if $\phi_{s,1} \leq \alpha_s$, then the curve of $\rho_s(\theta)$ has two pieces; otherwise, $\rho_s(\theta)$ is composed of one single piece.

Let $c'$ denote the intersection point between line segment $s$ and the circle $D_\Pi$ of radius $r$ centered at $b$.
In the case of $\phi_{s,1} \leq \alpha_s$, for any $\theta \in [\phi_{s,1}, \alpha_s]$, as depicted in Figure \ref{fig_in_a-d}B, $\rho_s(\theta)$ is given by the angle between $bt$ and $bc'$.
For any $\theta \in [\alpha_s, \phi_{s,2}]$ (and $\theta \in [\phi_{s,1}, \phi_{s,2}]$ in the case of $\phi_{s,1} > \alpha_s$), $\rho_s(\theta)$ is the angle between $bt$ and $bc'$, where $bc'$ intersects an endpoint of $s$ at an interior point of $bc'$ (see Figure \ref{fig_in_a-d}D).
Observe that $\rho_s(\theta) = 0$ when $\theta = \phi_{s,2}$.
A plot of function $\rho_s(\theta)$ is shown in Figure \ref{fig_in_plot}.

\begin{figure}[h]
	\centering
	\includegraphics[scale=0.24]{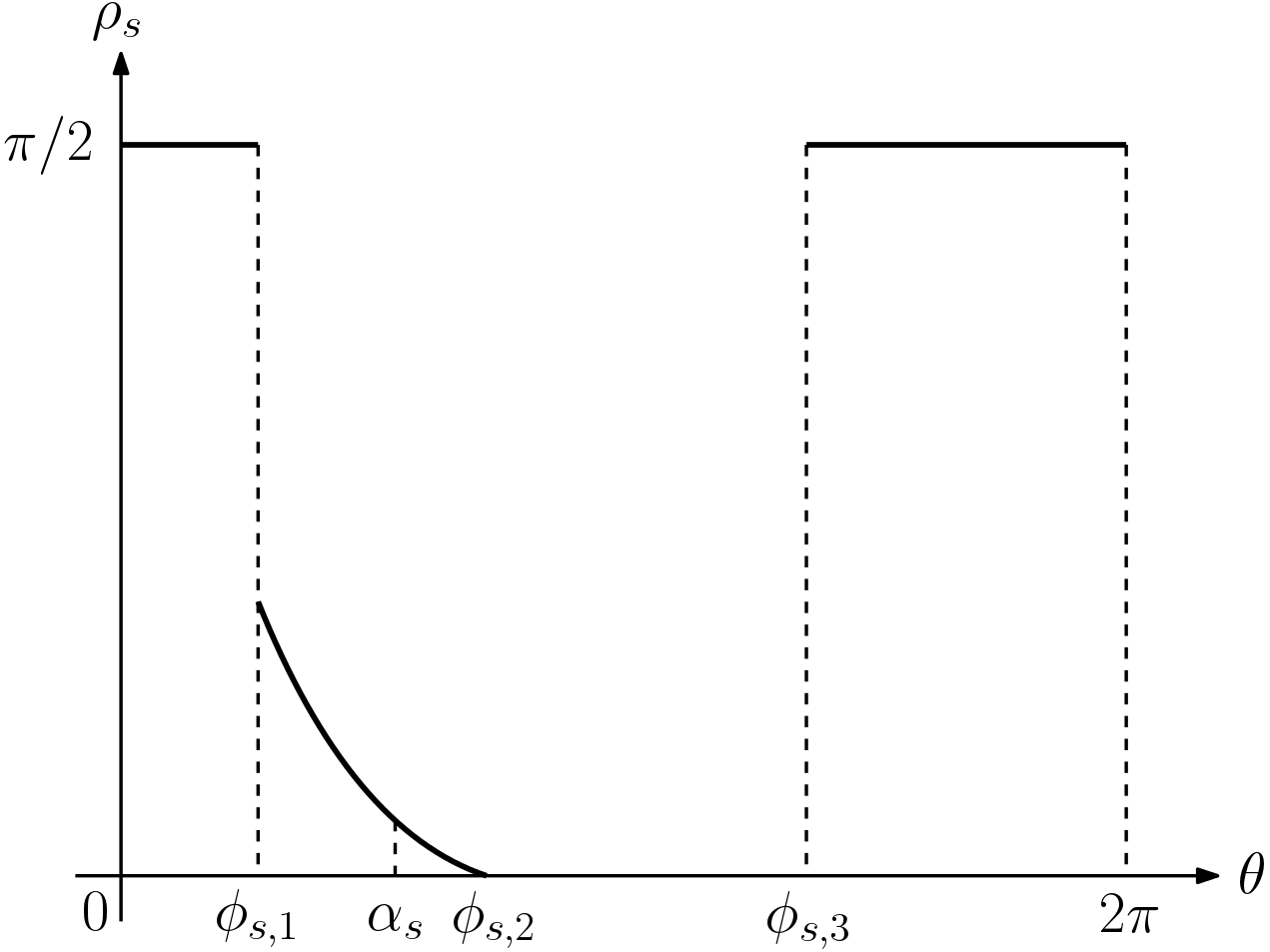}
	\caption{Illustration of function $\rho_s(\theta)$ in Case 1.}
	\label{fig_in_plot}
\end{figure}

\paragraph{Case 2: $s$ lies outside $C_\Pi$ and inside ${C_\Pi}'$}
For a line segment $s$ lying outside $C_\Pi$ and inside ${C_\Pi}'$, as depicted in Figure \ref{fig_out_a-d}, we only need to worry about computing $\rho_s$ for $\theta \in [\phi_{s,1}, \phi_{s,2}]$, given that $\rho_s= \pi/2$ for $\theta \in [0, \phi_{s,1}] \cup [\phi_{s,2}, 2\pi]$.

\begin{figure}[h]
	\centering
	\includegraphics[scale=0.18]{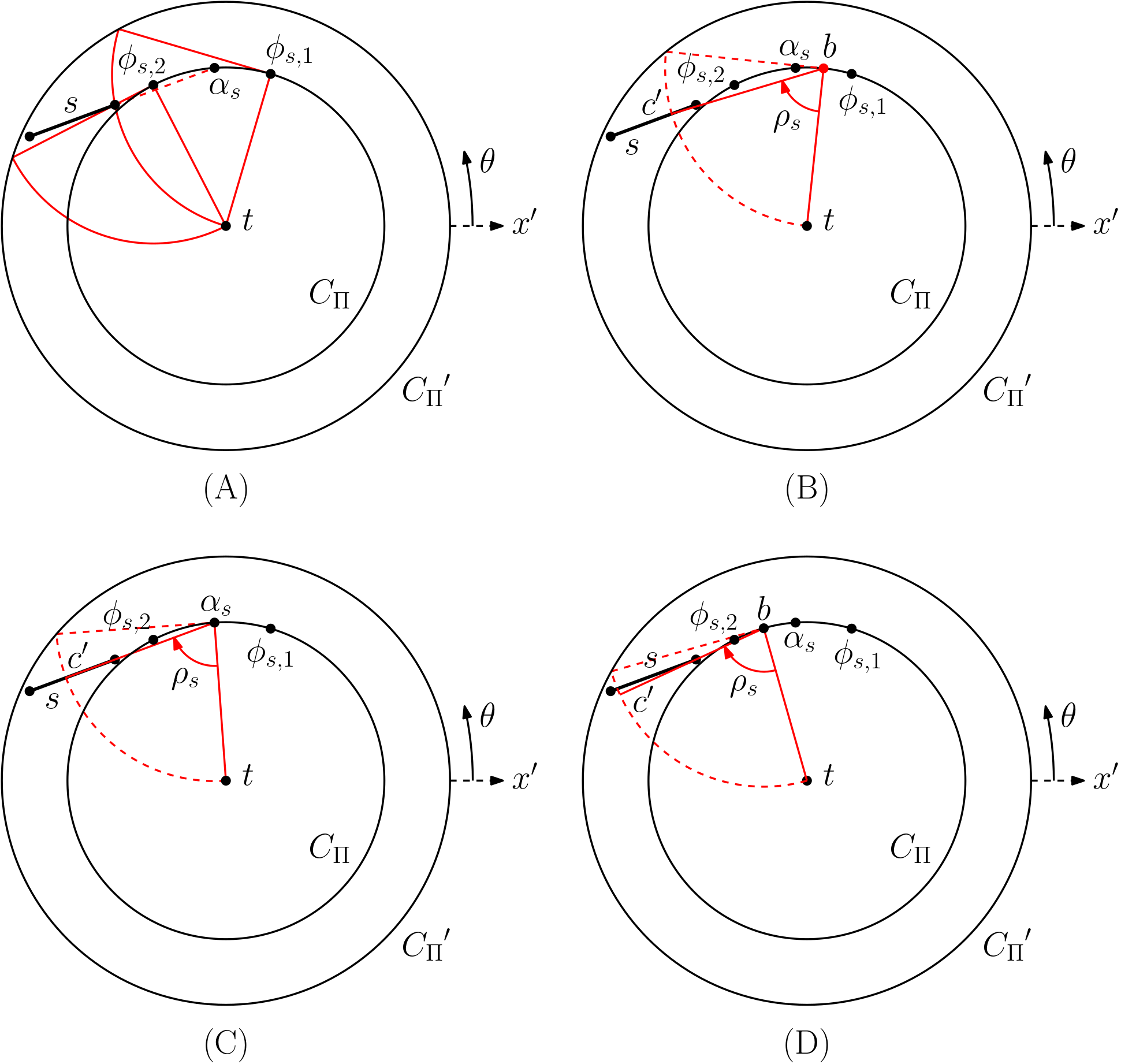}
	\caption{(A) Characterizing $\rho_s(\theta)$ for $\theta \in [\phi_{s,1}, \phi_{s,2}]$ in Case 2.
		Illustrations for $\rho_s(\theta)$ when (B) $\phi_{s,1} < \theta < \alpha_s$, (C) $\theta = \alpha_s$, and (D) $\alpha_s < \theta < \phi_{s,2}$, respectively.}
	\label{fig_out_a-d}
\end{figure}

\begin{figure}[h]
	\centering
	\includegraphics[scale=0.24]{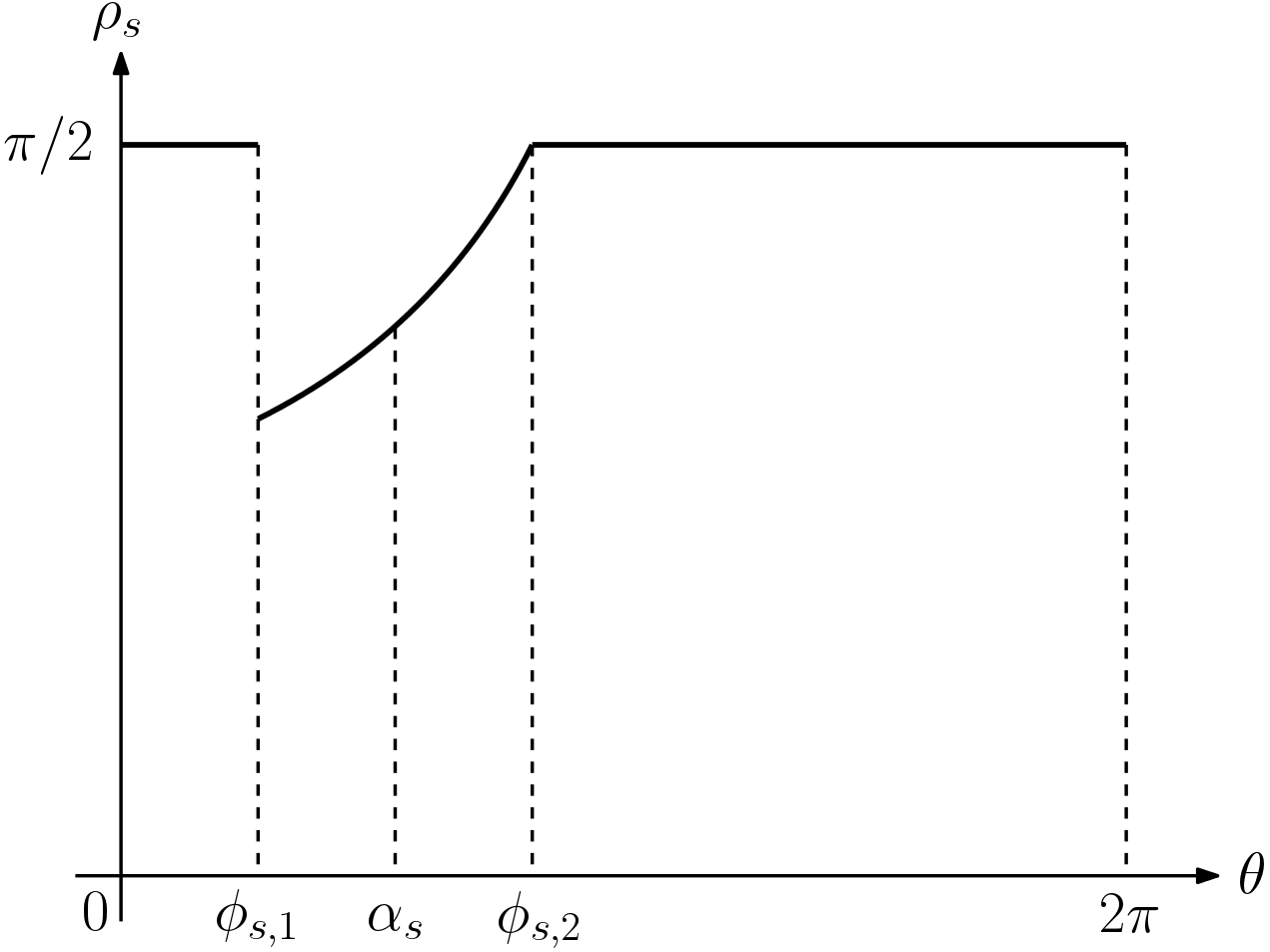}
	\caption{Illustration of function $\rho_s(\theta)$ in Case 2.}
	\label{fig_out_plot}
\end{figure}

For $\theta \in [\phi_{s,1}, \phi_{s,2}]$, $\rho_s(\theta)$ is characterized by a piecewise continuous curve, which consists of at most two pieces, corresponding to sub-intervals $[\phi_{s,1}, \alpha_s]$ and $[\alpha_s, \phi_{s,2}]$, where $\alpha_s$ is the angle $\theta$ of the intersection point, if any, between $C_\Pi$ and the supporting line of $s$.
Note that if $\alpha_s$ exists, then $\rho_s(\theta)$ has two pieces; otherwise, $\rho_s(\theta)$ is constituted of one single piece.

Notice that $\phi_{s,2}$ is defined differently from the previous case.
Here, $\phi_{s,2}$ is the largest $\theta$ at which the radius $bc'$ of the quart-sector of $\theta$ intersects with $s$ (at one of its endpoints).
In fact, when $\theta = \phi_{s,2}$, $\rho_s(\theta) = \pi/2$.
A plot of function $\rho_s(\theta)$ is shown in Figure \ref{fig_out_plot}. \\

Each of the curves $\rho_s(\theta)$ just described (for $\theta \in [\phi_{s,1}, \phi_{s,2}]$) is partially defined, continuous, and monotone over $\theta$.
Specifically, $\rho_s(\theta)$ is monotonically decreasing (resp. increasing) with respect to $\theta$ over the range of $[\phi_{s,1}, \phi_{s,2}]$ in Case 1 (resp. Case 2).
In fact, the curves behave like pseudo-line segments, since any two curves can only intersect at most once.

Observe that $\rho_s(\theta)$ is an inverse trigonometric function.
Nonetheless, we can define an algebraic function $f_s = \sin (\rho_s / 2)$ in terms of variables $x_b$ and $y_b$ (i.e., $x$- and $y$-coordinates of $b \in C_\Pi$).
The details of deriving $f_s$ are presented next.

\paragraph{Deriving algebraic function $f_s$}
\label{alg_func}
As one shall see, 
we can define an algebraic function $f_s$ based on $\rho_s(\theta)$.
Recall that, in Case 1, we are only concerned with characterizing $\rho_s(\theta)$ for  $\theta \in [\phi_{s,1}, \phi_{s,2}]$, which consists of at most two pieces, corresponding to two sub-intervals $[\phi_{s,1}, \alpha_s]$ and $[\alpha_s, \phi_{s,2}]$.
Let us consider the two sub-intervals individually.
Note that, in the following analysis, $\theta$ is represented by the $x$- and $y$-coordinates of $b \in C_\Pi$ (in order to obtain an algebraic expression).

\subparagraph{Sub-interval 1: $\theta \in [\phi_{s,1}, \alpha_s]$.}
For simplicity, let $\Pi$ be the $xy$-plane (in which line segment $s$ lies).
Without loss of generality, let $t$ be located at the origin.
Let $x_{c'}$ and $y_{c'}$ denote the $x$- and $y$-coordinates of $c'$, respectively.
Observe that
\begin{equation}
\label{eq8}
\sin \frac{\rho_s}{2} = \frac{\sqrt{{x_{c'}}^2 + {y_{c'}}^2}}{2r}
\end{equation}
Recall that $c'$ is the intersection point between line segment $s$ and the circle $D_\Pi$ of radius $r$ centered at $b$ (see Figure \ref{fig_alg_1}).
The coordinates of $c'$ can be expressed in terms of the coordinates of $b$ in the following manner.
Let $\ell$ be the supporting line of $s$.
Line $\ell$ can be represented by $y = mx + d$, where $m$ is the slope of $\ell$, and $d$ is the $y$-intercept of $\ell$.
Circle $D_\Pi$ can be formulated as $(x – x_b)^2 + (y – y_b)^2 = r^2$, where $x_b$ and $y_b$ are the $x$- and $y$-coordinates of $b$, respectively.
Then, $x_{c'}$ and $y_{c'}$ can be written as
\begin{gather}
\label{eq9}
x_{c'} = \frac{x_b + y_bm + dm \pm \sqrt{\delta}}{1 + m^2} \\
\label{eq10}
y_{c'} = \frac{d + x_bm + y_bm^2 \pm \sqrt{\delta}}{1 + m^2}
\end{gather}
where $\delta = r^2 (1 + m^2) - (y_b - mx_b - d)^2$.
Let $f_s$ = $\sin (\rho_s / 2)$ (i.e., Equation \ref{eq8}).
Note that $f_s \in [0, \sqrt{1/2}]$.
Clearly, by substituting Equations \ref{eq9} and \ref{eq10} into Equation \ref{eq8}, $f_s$ can be expressed algebraically as a function of $x_b$ and $y_b$, where $b \in C_\Pi$ and ${x_b}^2 + {y_b}^2 = r^2$.

\begin{figure}[h]
	\centering
	\includegraphics[scale=0.19]{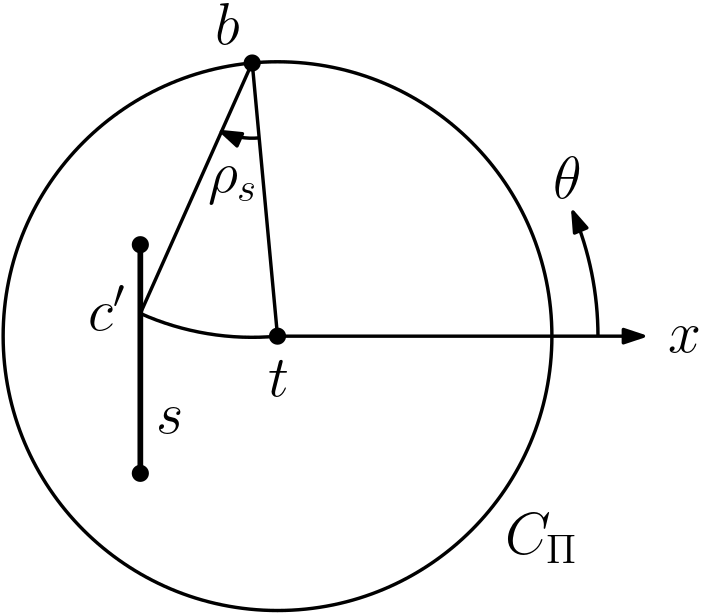}
	\caption{Case for $\theta \in [\phi_{s,1}, \alpha_s]$.}
	\label{fig_alg_1}
\end{figure}

\subparagraph{Sub-interval 2: $\theta \in [\alpha_s, \phi_{s,2}]$.}
Observe that Equation \ref{eq8} also holds true for this sub-interval.
Let $p$ be the intersection point between $s$ and $bc'$ (see Figure \ref{fig_alg_2}).
Note that $p$ is an endpoint of $s$.
Let $x_p$ and $y_p$ denote the $x$- and $y$-coordinates of $p$.
Then, $x_{c'}$ and $y_{c'}$ can be expressed as
\begin{gather}
\label{eq11}
x_{c'} = x_b + \sqrt{\frac{r^2}{1 + (\frac{y_p - y_b}{x_p - x_b})^2}} \\
\label{eq12}
y_{c'} = y_b + \sqrt{\frac{r^2}{1 + (\frac{y_p - x_b}{y_p - y_b})^2}}
\end{gather}
As before, we let $f_s$ = $\sin (\rho_s / 2)$, and $f_s$ can be expressed algebraically as a function of $x_b$ and $y_b$ by substituting Equations \ref{eq11} and \ref{eq12} into Equation \ref{eq8}. \\

\begin{figure}[h]
	\centering
	\includegraphics[scale=0.19]{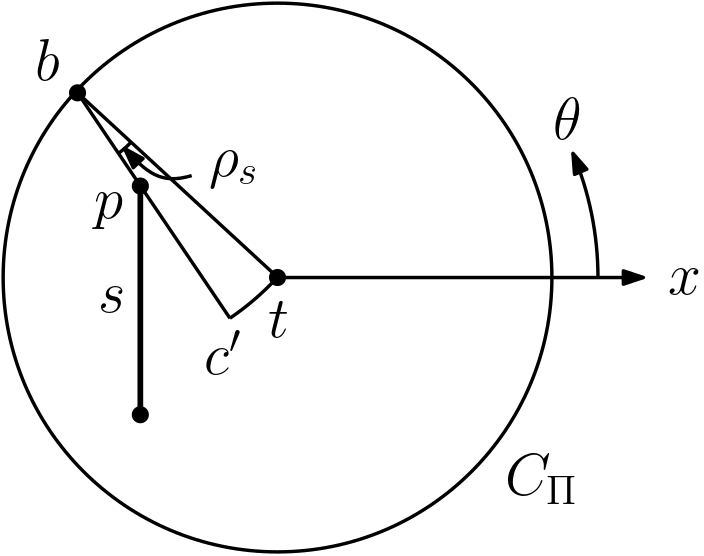}
	\caption{Case for $\theta \in [\alpha_s, \phi_{s,2}]$.}
	\label{fig_alg_2}
\end{figure}

A similar derivation of $f_s$ can also be performed in Case 2.
Since $\rho_s$ is partially defined, continuous, and monotone over $\theta \in [\phi_{s,1}, \phi_{s,2}]$, so must be function $f_s = \sin (\rho_s / 2)$. \\

At this point, we are in position to claim a new unified data structure for solving the circular sector emptiness query problem in 2D.
This $\mathbb{R}^2$ data structure simplifies the two-step approach described in \cite{teo20traj} (which consists of a circular arc intersection query and a circular sector emptiness query) while having the same time and space complexity.
Specifically, we construct two lower envelopes $V_1$ and $V_2$ of the piecewise algebraic curves $f_s$ for all given line segments $s \in P$ as follows: $V_1$ is the lower envelope of the curves $f_s$ for $y_b \geq 0$ (i.e., $0 \leq \theta \leq \pi$), whereas $V_2$ is for $y_b < 0$ (i.e., $\pi < \theta < 2\pi$).
Note that $V_1$ and $V_2$ are computed as functions of $x_b$ only, given that $y_b = + (r^2 – {x_b}^2)^{1/2}$ for $y_b \geq 0$, and $y_b = - (r^2 – {x_b}^2)^{1/2}$ for $y_b < 0$.

Since each pair of curves $\rho_s$ intersect in at most one point, any two curves $f_s$ do too.
Thus, the size of the lower envelope is bounded by the third-order Davenport-Schinzel sequence, whose length is at most $O(n \alpha(n))$, where $\alpha(n)$ is the inverse Ackermann function.
The lower envelope can be computed in $O(n \log n)$ time \cite{Hersh89find,sharir95dav}.

Given a query circular sector $\sigma$, let $b_\sigma$ be the apex of $\sigma$.
If $y_{b_\sigma} \geq 0$, then $x_{b_\sigma}$ is looked up in $V_1$ by using a binary search, which takes $O(\log n)$ time; otherwise, $x_{b_\sigma}$ is looked up in $V_2$.
Let $\rho$ be the acute angle between the two bounding radii of $\sigma$.
If $\sin (\rho / 2)$ is less than $f_s(x_b)$ for all line segments $s \in P$, then $\sigma$ does not intersect $P$.
Hence, we have the following result.

\begin{thm}
A set $P$ of $n$ line segments in $\mathbb{R}^2$ can be preprocessed in $O(n \log n)$ time into a data structure of size $O(n \alpha(n))$ so that, for a query circular sector $\sigma$ with a fixed radius $r$ and a fixed arc endpoint $t$, one can determine if $\sigma$ intersects $P$ in $O(\log n)$ time.
\end{thm}

As it turns out, we do not need to explicitly describe the function of the lower envelope.
In fact, it suffices to have an implicit representation of the lower envelope -- specifically, its vertices and the line segments that induce its various pieces.
Generally, in order to obtain an implicit representation of the lower envelope (of a given set of curves), we require i) the number of times each pair of curves intersect, ii) the ability to compute the intersection points between any two given curves, and iii) the ability to determine whether one curve lies above or below another to the left and right of their intersection.
Consider the following approach for computing an implicit representation of the lower envelope.


\paragraph{Computing implicit lower envelopes}
\label{imp_low}
We only address the case of computing $V_1$ ($0 \leq \theta \leq \pi$), and the other case of $V_2$ ($\pi < \theta < 2\pi$) can be handled symmetrically.
Recall that each pair of curves $f_s$ intersect at most once.
Thus, as with computing the lower envelope of a set of $n$ lines segments using a divide-and-conquer algorithm, we first sort the curves $f_s$ by $\phi_{s,1}$.
Note that, for a given line segment $s$, $\phi_{s,1}$ can be computed without knowing $f_s$, since $\phi_{s,1}$ corresponds to the smallest angle $\theta$ at which the circle of radius $r$ centered at $b \in C_\Pi$ intersects with $s$.
We divide the set of curves $f_s$ into two equal sets by $\phi_{s,1}$, recursively compute the lower envelope of each set, and then merge the two lower envelopes to obtain the final result.
The intersection of any two curves $f_{s_i}$ and $f_{s_j}$ (which correspond to line segments $s_i$ and $s_j$, respectively) can be determined as follows.
There are two cases to be considered, depending on how $bc'$ is supported by $s_i$ and $s_j$.

\subparagraph{Case A.}
Line segment $bc'$ intersects $s_i$ and $s_j$ at their endpoints (see Figure \ref{fig_im_a}).
Note that the corresponding $x_b$ can be computed algebraically, given that $b$ is the intersection point between $C_\Pi$ and $bc'$.
Specifically, $C_\Pi$ can be represented as $x^2 + y^2 = r^2$.
Let $p$ and $q$ denote the endpoints of $s_i$ and $s_j$, respectively, that intersect $bc'$.
Then, the line that passes through $p$ and $q$ can be expressed as $y = m_{pq} x + d_{pq}$, where $m_{pq} = (y_q - y_p) / (x_q - x_p)$ and $d_{pq} = y_p - [(y_q - y_p)/(x_q - x_p)] x_p$.
Hence,
\[ x_b = \frac{-d_{pq} m_{pq} \pm \sqrt{\delta_{pq}}}{1 + {m_{pq}}^2} \]
where $\delta_{pq} = r^2 (1 + {m_{pq}}^2) - {d_{pq}}^2$.
We can also decide in constant time (based on the slope of $bc'$) which of the two curves $f_{s_i}$ and $f_{s_j}$, to the left and right of their intersection at $x_b$, lies above/below the intersection point.
For instance, in the example illustrated in Figure \ref{fig_im_a}, we can easily determine that $f_{s_j}$ lies below (resp. above) $f_{s_i}$ to the left (resp. right) of their intersection at $x_b$.

\begin{figure}[h]
	\centering
	\includegraphics[scale=0.19]{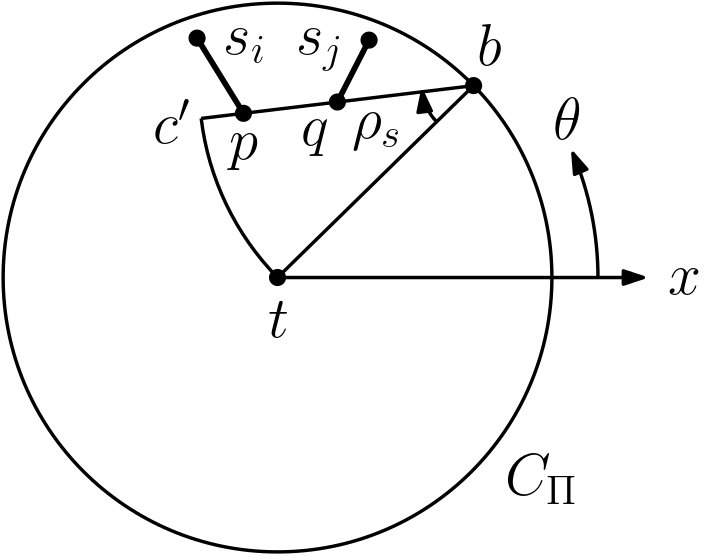}
	\caption{Case A.}
	\label{fig_im_a}
\end{figure}

\subparagraph{Case B.}
Line segment $bc'$ intersects $s_i$ such that $c'$ coincides with an interior point of $s_i$, and intersects $s_j$ at its endpoint (see Figure \ref{fig_im_b}).
As with the previous case, we can compute the corresponding $x_b$ algebraically.
Let $\ell$ denote the supporting line of $s_i$, and let $q$ be the intersection point between $s_j$ and $bc'$.
Line $\ell$ can be expressed as $y = m_\ell x + d_\ell$.
Observe that the line passing through $c'$ and $q$ intersects $C_\Pi$ at $b$; let the line be represented as $y = m_{c'q} x + d_{c'q}$.
Then,
\begin{gather}
\label{eq13}
x_b = \frac{-d_{c'q} m_{c'q} \pm \sqrt{\delta_{c'q}}}{1 + {m_{c'q}}^2} \\
\label{eq14}
y_b = \frac{d_{c'q} \pm m_{c'q} \sqrt{\delta_{c'q}}}{1 + {m_{c'q}}^2}
\end{gather}
where $m_{c'q} = (y_q - y_{c'}) / (x_q - x_{c'})$, $d_{c'q} = y_{c'} - [(y_q - y_{c'})/(x_q - x_{c'})] x_{c'}$, and $\delta_{c'q} = r^2 (1 + {m_{c'q}}^2) - {d_{c'q}}^2$.  Since the circle of radius $r$ centered at $b$ intersects $\ell$ at $c$,
\begin{gather}
\label{eq15}
x_{c'} = \frac{x_b + y_b m_\ell + d m_\ell \pm \sqrt{\delta_\ell}}{1 + {m_\ell}^2} \\
\label{eq16}
y_{c'} = \frac{d_\ell + x_b m_\ell + y_b {m_\ell}^2 \pm \sqrt{\delta_\ell}}{1 + {m_\ell}^2}
\end{gather}
where $\delta_\ell = r^2 (1 + {m_\ell}^2) - (y_b - {m_\ell} x_b - d_\ell)^2$.
Clearly, we can determine $x_b$ by using Equations \ref{eq13}--\ref{eq16}.
As before, we can also decide in $O(1)$ time which of the two curves $f_{s_i}$ and $f_{s_j}$, to the left and right of their intersection at $x_b$, lies above/below the intersection point. \\

\begin{figure}[h]
	\centering
	\includegraphics[scale=0.19]{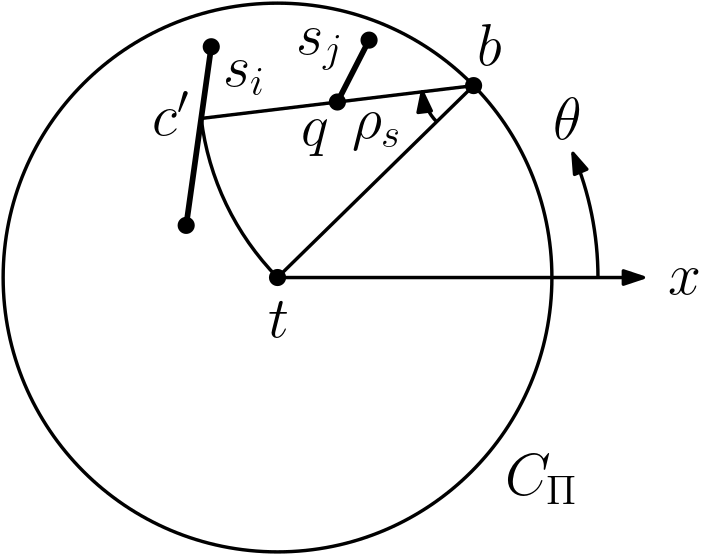}
	\caption{Case B.}
	\label{fig_im_b}
\end{figure}

Notice that we did not compute a full (functional) description of the lower envelope.
Instead, for each piece of the lower envelope, we store the information about the line segment $s$ associated with the curve $f_s$ to which the piece of the lower envelope belongs.

At query time, given a circular sector $\sigma$ (and its associated $x_b$), we can determine if $\sigma$ intersects any line segment in $P$, without having to know the equation for the curve $f_s$ defining the lower envelope at $x_b$, as follows.
We retrieve the line segment $s$ that induces the curve $f_s$ defining the lower envelope at $x_b$, and check if $bc$ of $\sigma$ intersects $s$.
If $bc$ of $\sigma$ does not intersect $s$, then $\sigma$ does not intersect $P$. \\

We now continue to derive our result for circular sector emptiness queries in 3D.
For each triangle $\tau \in P$, we define a trivariate function $\rho_\tau(I, \Omega, \theta)$, so that $\rho_s(\theta)$ is characterized with respect to all planes $(I, \Omega)$.
As with the two-dimensional case, $\rho_\tau(I, \Omega, \theta)$ is an inverse trigonometric function, based on which we can define an algebraic function $f_\tau$ (of constant degree) in terms of four variables as follows.

\paragraph{Deriving algebraic function $f_\tau$}
\label{alg_func_3d}
In the following analysis, Cartesian variables are used in place of $I$, $\Omega$, and $\theta$, in order to derive an algebraic expression for $f_\tau$.
As with characterizing the univariate function $\rho_s(\theta)$, there are two cases to be considered, depending on whether a given triangle $\tau$ is located 1) inside sphere $C$ or 2) outside sphere $C$ and inside sphere $C'$.
Given the similarity in analysis between the two cases, we only present the details for Case 1.
As before, we are concerned with computing $\rho_\tau$ for two contiguous sub-intervals of $(I, \Omega, \theta)$.
In one of the sub-intervals, $bc'$ always intersects the interior of $\tau$ at $c'$.
In the other sub-interval, $bc'$ intersects an edge of $\tau$.
Let us examine the two sub-intervals individually.

\subparagraph{Sub-interval 1: $bc'$ intersects the interior of $\tau$ at $c'$.}
Without loss of generality, let $t$ be located at the origin.
Let $x_{c'}$, $y_{c'}$, and $z_{c'}$ denote the $x$-, $y$-, and $z$-coordinates of $c'$, respectively.
Observe that
\begin{equation}
\label{eq17}
f_\tau = \sin \frac{\rho_\tau}{2} = \frac{\sqrt{{x_{c'}}^2 + {y_{c'}}^2 + {z_{c'}}^2}}{2r}
\end{equation}
Let $G$ be the plane passing through points $t$ and $b$.
Plane $G$ can be represented as $g_1 x + g_2 y + g_3 z + g_4 = 0$, where $g_4 = 0$ and $g_3 = (- g_1 x_b - g_2 y_b) / z_b$.  If we define $g = g_1 / g_2$, then the expression for $G$ becomes $gx + y + [(- g_1 x_b - g_2 y_b) / z_b] z = 0$.
Given that $G$ contains $c'$,
\begin{equation}
\label{eq18}
gx_{c'} + y_{c'} + \left(\frac{- g_1 x_b - g_2 y_b}{z_b}\right) z_{c'} = 0
\end{equation}
Let $H$ be the supporting plane of $\tau$.
Plane $H$ can be expressed as $h_1 x + h_2 y + h_3 z + h_4 = 0$, where $h_1$, $h_2$, $h_3$, and $h_4$ are the known parameters determined based on the three vertices of $\tau$.
Since $H$ contains $c'$,
\begin{equation}
\label{eq19}
h_1 x_{c'} + h_2 y_{c'} + h_3 z_{c'} + h_4 = 0
\end{equation}
Notice that $|bc'| = r$.
Thus,
\begin{equation}
\label{eq20}
(x_{c'} - x_b)^2 + (y_{c'} - y_b)^2 + (z_{c'} - z_b)^2 = r^2
\end{equation}
Based on the three Equations \ref{eq18}, \ref{eq19}, and \ref{eq20}, we can obtain an algebraic expression for $x_{c'}$, $y_{c'}$, and $z_{c'}$, respectively, in terms of variables $g$, $x_b$, $y_b$, and $z_b$.
By substituting the resulting algebraic expressions for $x_{c'}$, $y_{c'}$, and $z_{c'}$ into Equation \ref{eq17},  $f_\tau$ can be expressed algebraically as a function (of degree one) of $g$, $x_b$, $y_b$, and $z_b$, where $b \in C$ and ${x_b}^2 + {y_b}^2 + {z_b}^2 = r^2$.

\subparagraph{Sub-interval 2: $bc'$ intersects an edge of $\tau$.}
Let $s$ denote the edge of $\tau$ intersected by $bc'$.
Edge $s$ can be expressed in parametric form as
\begin{gather*}
x = (1 - \lambda_s) x_u + \lambda_s x_v \\
y = (1 - \lambda_s) y_u + \lambda_s y_v \\
z = (1 - \lambda_s) z_u + \lambda_s z_v
\end{gather*}
where $u$ and $v$ are the two endpoints of $s$, and $0 \leq \lambda_s \leq 1$.
Let $p$ denote the intersection point between $bc'$ and $s$.
The supporting line $\ell$ of $bc'$ can be represented as
\begin{gather*}
x = (1 - \lambda_\ell) x_b + \lambda_\ell x_p \\
y = (1 - \lambda_\ell) y_b + \lambda_\ell y_p \\
z = (1 - \lambda_\ell) z_b + \lambda_\ell z_p
\end{gather*}
where $\lambda_\ell \in \mathbb{R}$.
Given that $p$ lies in $s$, and $\ell$ contains $c'$,
\begin{gather}
\label{eq21}
x_{c'} = (1 - \lambda_\ell) x_b + \lambda_\ell [(1 - \lambda_s) x_u + \lambda_s x_v] \\
\label{eq22}
y_{c'} = (1 - \lambda_\ell) y_b + \lambda_\ell [(1 - \lambda_s) y_u + \lambda_s y_v] \\
\label{eq23}
z_{c'} = (1 - \lambda_\ell) z_b + \lambda_\ell [(1 - \lambda_s) z_u + \lambda_s z_v]
\end{gather}
Note that Equation \ref{eq20} still applies in this case.
Substitute Equations \ref{eq21}, \ref{eq22}, and \ref{eq23} into Equation \ref{eq20}, and solve for $\lambda_\ell$ in respect of variables $\lambda_s$, $x_b$, $y_b$, and $z_b$.
Then, by substituting the resulting expression for $\lambda_\ell$ into Equations \ref{eq21}, \ref{eq22}, and \ref{eq23}, we obtain an algebraic expression (of degree one) for $x_{c'}$, $y_{c'}$, and $z_{c'}$, respectively, in terms of variables $\lambda_s$, $x_b$, $y_b$, and $z_b$.

At last, by substituting the algebraic expressions for $x_{c'}$, $y_{c'}$, and $z_{c'}$ into Equation \eqref{eq17},  $f_\tau$ can be expressed algebraically as a function of $\lambda_s$, $x_b$, $y_b$, and $z_b$, where $b \in C$ and ${x_b}^2 + {y_b}^2 + {z_b}^2 = r^2$.

\paragraph{Finishing up}
We construct two lower envelopes $V_1$ and $V_2$ of the piecewise algebraic functions $f_\tau$ for all given triangles $ \tau \in P$, such that $V_1$ is the lower envelope of $f_\tau$ for $z_b \geq 0$, and $V_2$ is for $z_b < 0$.
As a result, $V_1$ and $V_2$ are functions of only three variables.

For constructing the lower envelope of the trivariate piecewise algebraic functions (of constant description complexity) just described, the best-known performance bounds are given by Koltun \cite{koltun04almost} and Agarwal et al. \cite{agarwal97computing}.
Specifically, according to Koltun \cite{koltun04almost}, we can compute the lower envelope deterministically in $O(n^{4+\epsilon})$ time and store it in a data structure of size $O(n^{4+\epsilon})$ and query time $O(\log n)$, for any $\epsilon > 0$.
On the other hand, Agarwal et al. \cite{agarwal97computing} showed that the lower envelope can be constructed in randomized expected time $O(n^{3+\epsilon})$ and stored in an $O(n^{3+\epsilon})$-size data structure with a query time of $O(\log^2 n)$.
Thus, we arrive at the result below.

\begin{thm}
\label{lem5}
For any constant $\epsilon > 0$, a set $P$ of $n$ triangles in $\mathbb{R}^3$ can be preprocessed in $O(n^{4+\epsilon})$ time into a data structure of size $O(n^{4+\epsilon})$ so that, for a query circular sector $\sigma$ with a fixed radius $r$ and an endpoint of its arc located at fixed point $t$, one can determine whether $\sigma$ intersects $P$ in $O(\log n)$ time.
\end{thm}

Given that $O(n^4)$ queries are to be processed in the worst case, the following result is obtained.

\begin{lem}
\label{lem6}
A feasible articulated probe trajectory, if one exists, can be determined in $O(n^{4+\epsilon})$ time using $O(n^{4+\epsilon})$ space, for any constant $\epsilon > 0$.
\end{lem}

Since the space/time complexity of finding an extremal feasible articulated trajectory (Lemma \ref{lem6}) is dominant over that of the case of unarticulated trajectory (Lemma \ref{lem3}), the final result can be stated as follows.

\begin{thm}
One can determine if a feasible trajectory exists, and if so, report (at least) one such trajectory in $O(n^{4+\epsilon})$ time using $O(n^{4+\epsilon})$ space, for any constant $\epsilon > 0$.
\end{thm}

As an alternative, an $O(n^5)$-time algorithm with linear space usage is achievable by performing a simple $O(n)$-check on each of the $O(n^4)$ extremal trajectories.
The proposed enumeration algorithm is easy to implement and could be quite fast in practice, since the enumeration stops once a feasible solution is found.
In addition, given that an extremal feasible probe trajectory amid $n$ line segment obstacles in the plane can be obtained in $O(n^2)$ time \cite{daescu20char}, one could use a sampling-based method to find a feasible probe trajectory among $n$ triangular obstacles in 3D space.
Specifically, one could generate a finite set of planes passing through fixed point $t$ either randomly or using an appropriate discretization scheme (i.e., with respect to parameters $I$ and $\Omega$), and then employ the existing data structure from \cite{teo20traj} to compute the set of feasible trajectories, if any, in each plane.

\section{Conclusion}
\label{conc}
We have presented efficient data structures and algorithms for solving a trajectory planning problem involving a simple articulated probe in 3D space.
In particular, we have shown that a feasible probe trajectory, among $n$ triangular obstacles, can be found in $O(n^{4+\epsilon})$ time, for any constant $\epsilon > 0$.
In the process, we have solved a special case of the circular sector emptiness query problem in 3D and simplified the corresponding data structure in 2D.
We leave open the following questions:
1) Since our approach is enumerative, can we speed up the process of finding one feasible solution?
2) Is it possible to extend the current algorithm to finding feasible probe trajectories of a given (or maximum) clearance?

\section*{Acknowledgment}
The authors would like to thank Pankaj K. Agarwal for helpful discussions and comments.

\bibliographystyle{elsarticle-harv} 
\bibliography{refs}

\end{document}